\documentclass[12pt,letterpaper]{article}

\usepackage{renstyles}

\usepackage{pgf}
\usepackage{tikz}
\usepackage{subcaption}
\usepackage{multirow}
\usepackage{array}
\usepackage{algpseudocode}
\algnewcommand\algorithmicswitch{\textbf{Select}}
\algnewcommand\algorithmiccase{\textbf{If model} ==}
\algnewcommand\algorithmicassert{\texttt{assert}}
\algnewcommand\Assert[1]{\State \algorithmicassert(#1)}%
\algdef{SE}[SWITCH]{Switch}{EndSwitch}[1]{\algorithmicswitch\ #1\ \algorithmicdo}{\algorithmicend\ \algorithmicswitch}%
\algdef{SE}[CASE]{Case}{EndCase}[1]{\algorithmiccase\ #1}{\algorithmicend\ \algorithmiccase}%
\algtext*{EndSwitch}%
\algtext*{EndCase}%

\renewcommand{\RED}[1]{#1}

\title{Phase retrieval via media diversity}
\author{
    Yan Cheng\thanks{Department of Applied Physics and Applied Mathematics, Columbia University, New York, NY 10027; yc3855@columbia.edu}
    \and
    Kui Ren \thanks{Department of Applied Physics and Applied Mathematics, Columbia University, New York, NY 10027; kr2002@columbia.edu}
    \and
    Nathan Soedjak\thanks{Department of Applied Physics and Applied Mathematics, Columbia University, New York, NY 10027; ns3572@columbia.edu}
}
\begin{document}


\maketitle

\begin{abstract}

\RED{This work studies phase retrieval in the Fresnel regime, aiming to recover the phase of an incoming wave from multi-plane intensity measurements behind different types of linear and nonlinear media. We show that unique phase retrieval can be achieved by utilizing intensity data produced by multiple media. This uniqueness does not require prescribed boundary conditions for the phase in the incidence plane, in contrast to existing phase retrieval methods based on the transport of intensity equation. Moreover, the uniqueness proofs lead to explicit phase reconstruction algorithms. Numerical simulations are presented to validate the theory.}

\end{abstract}




\section{Introduction}

Phase retrieval is an important task in many areas of science and technology~\cite{Fienup-AO82,OpLi-IEEE81,ShElCoChMiSe-IEEE15}. There has been tremendous development in phase retrieval methods in the past decades; see, for instance, ~\cite{BaWo-JOSA03,ChBiZh-Elight21,DoVaMaPhGiUn-IEEE23,MaFaKo-AO17,PaGoKh-AO16,ScGuPaPa-PRA11,WaSoWaReZhDoDiBaZhZhLa-LSA24,WiSuArRoTr-OE24,YoOnSaKoNo-OE21,ZhTiShLuLi-OE19,ZhZhMaFaZh-MST20,ZuLiSuFaZhLuZhWaHuCh-OLE20} and references therein for some random samples of recent progress. However, researchers are still constantly pursuing new improvements over the state-of-the-art of phase retrieval methods.

Besides the advances in numerical reconstruction algorithms~\cite{CaLiSo-IEEE15,FaSt-AN20,KlSaTi-IP95,ZhZhMaFaZh-MST20,ZuLiSuFaZhLuZhWaHuCh-OLE20}, a significant amount of efforts in the area have been (necessarily) spent on designing measurement strategies to enrich the diversity of the phase information encoded in the intensity data from which phase data are extracted. Examples of such strategies include, but are not limited to, for instance, measuring intensity data at multiple transverse planes (i.e. multi-plane measurement)~\cite{GuFi-OE08,GuJiSoWaShZhZh-BOE21}, exploiting phase-amplitude interactions due to nonlinear propagation~\cite{LuBaFl-FO11,LuBaWiKuFl-AO13,PaNu-OL02}, enhancing partial coherency via structured multi-frequency, multi-probe, or multi-illumination measurements~\cite{HeLiZh-OL18,KaPaKoSiEi-IEEE20,KoPaKnEi-IEEE21,PaKnEi-IEEE17,PoYaBo-JFAA14}, and many more~\cite{ArLaLa-IEEE18,DeBo-JOSA03,EgMi-OL19,RoFa-APL04}. 

To be concrete, let us consider the propagation of an incident wave field $f(\bx)$ through a medium with refractive index $n$. In the parabolic regime, this propagation is modeled by the Schr\"{o}dinger equation:
\begin{equation}\label{EQ:LS}
    \begin{array}{rcll}
	i\dfrac{\partial u}{\partial z} +\kappa(z,\bx)\Delta u &=& 0, & \mbox{in}\ \ (0, L_z]\times \Omega\\[2ex]
	u(0, \bx)&=&f(\bx),& \mbox{in}\ \ \Omega
	\end{array}
\end{equation}
where the propagation is along the positive $z$ direction, $\Delta$ is the standard Laplace operator on the variable $\bx\in\Omega\subset\bbR^{d}$ in the transverse plane, and $\kappa:=\dfrac{\lambda}{4 \pi  n(z,\bx)}$ is the inverse refractive index of the medium with $\lambda$ the free-space wavelength of the incoming signal. 

When $\Omega$ is a bounded region, we need to supplement equation~\eqref{EQ:LS} with appropriate boundary conditions, depending on the experimental setup. In the numerical simulations we present in this paper (in~\Cref{SEC:Numer} and~\Cref{SEC:Generalizations}), we impose periodic boundary conditions. For the main part of this work, the specific boundary conditions we impose do not play a critical role.

When the refractive index $n$ is constant and $\Omega=\bbR^{d}$, the intensity of the output signal at any $0<z\le L_x$ is given as~\cite{Cazenave-Book03}
\begin{equation}
    \left |u(z,\bx)\right|=\Big(\frac{n}{\lambda z}\Big)^{\frac{d}{2}} \Big|\int_{\bbR^d}e^{i\frac{n\pi}{\lambda z}|\bx-\by|^2} f(\by) d\by\Big|\,.
\end{equation}
The classical phase retrieval problem, with transverse plane diversity, in this setting is to reconstruct $f$ from intensity data $|u(z,\bx)|$ measured at different transverse planes located in $(0, L_z)$~\cite{GuFi-OE08,GuJiSoWaShZhZh-BOE21}.

The aim of this work is to systematically study the use of media diversity to enhance phase retrieval~\cite{BaWaFl-NP09,LuBaFl-FO11}. That is, instead of measuring the intensity of the incoming wave in one medium, we measure the intensity of the field after it, after signal splitting, passes through a group of different media, parameterized by the refractive index $n$ in the linear Schr\"{o}dinger model~\eqref{EQ:LS}. Moreover, we are interested in utilizing media nonlinearity to further enrich the information content of intensity data. This is motivated by extensive studies in recent years, which show that the nonlinearity of media can create a useful interaction between the phase and the amplitude of the underlying wave and thus improve classical linear phase retrieval methods~\cite{BrGuFi-OE09,LuBaWiKuFl-AO13, LuLuFl-OE16,PaNu-OL02,TaPaRoMo-PRE03}. 
As an example, we consider the propagation of an incident wave field $f(\bx)$ through a medium that exhibits quadratic nonlinearity. The propagation is described by the following quadratic nonlinear Schr\"{o}dinger equation for the wave field $u$~\cite{BeTa-JFA06,IwOg-TAMS15,KePoVe-TAMS96,LaLuZh-arXiv23}:
\begin{equation}\label{EQ:QNLS}
    \begin{array}{rcll}
	i\dfrac{\partial u}{\partial z} +\kappa(z,\bx) \Delta u +\beta(z, \bx) u^2 &=& 0, & \mbox{in}\ \ (0, L_z]\times \Omega\\[2ex]
	u(0, \bx)&=&f(\bx),& \mbox{in}\ \ \Omega
	\end{array}
\end{equation}
which is sometimes used as a simplified model for media with the second harmonic generation. The coefficient $\beta$ is related to the second harmonic generation of the medium.

The general phase retrieval question we want to ask is: assuming that we can measure the intensity of the wave field $u$ after it passes through the medium, at a sequence of $N_p$ transverse planes located at $\{z_p\}_{p=1}^{N_p}$, for a group of $N_m$ different media $\{\kappa_m, \beta_m\}_{m=1}^{N_m}$, can we recover the incident field $f(\bx)$? Mathematically, this question can be formulated on the continuous level as an inverse problem as follows.

{\bf Phase retrieval problem:} Determine $f(\bx)$ from the data encoded in the map
\begin{equation}\label{EQ:Data Map}
    \Lambda_f^{[Z_1, Z_2]}: 
    \begin{array}{rcl}
    (\kappa, \beta) &\mapsto& |u(z, \bx)|_{[Z_1, Z_2]\times \Omega}\\
    \cP\times \cP & \to & \cD
    \end{array}
\end{equation}
where $0\le Z_1<Z_2\le L_z$, while $\cP$ and $\cD$ denote respectively the parameter (i.e., medium) and data spaces of the problem.

Various recent studies have shown that such phase reconstructions are sometimes possible; see, for instance, ~\cite{BrGuFi-OE09,LuBaWiKuFl-AO13, LuLuFl-OE16,PaNu-OL02,TaPaRoMo-PRE03} and references therein. In particular, for nonlinearities modelled by the Gross-Pitaevskii (GP) equation~\cite{Berge-PR98,Fibich-Book15}: 
\begin{equation}\label{EQ:GP}
    \begin{array}{rcll}
	i\dfrac{\partial u}{\partial z} + \kappa(z, \bx) \Delta u + \beta(z, \bx)|u|^2 u &=& 0, &  \mbox{in}\ \ (0, L_z]\times \Omega\\[2ex]
	u(0, \bx)&=&f(\bx),& \mbox{in}\ \ \Omega
    \end{array}
\end{equation}
computational and experimental results~\cite{BaWaFl-NP09,GoPs-PRA11,LuBaWiKuFl-AO13, LuLuFl-OE16,LuBaFl-FO11} demonstrated that nonlinearity increased the robustness of phase retrieval under different scenarios.

It is well-known that, in the case of linear media, $u e^{i\varphi}$ with any constant phase $\varphi$ is also a solution to the linear Schr\"odinger equation~\eqref{EQ:LS}. In other words, the equation is invariant under the transform $u\mapsto ue^{i\varphi}$. Therefore, if we only measure the intensity of the wave field, we can, at most, have the unique reconstruction of the phase of $f$ up to a constant shift. This is true even in the case of the nonlinear model~\eqref{EQ:GP} as the same invariance holds for the solution to the Gross-Pitaevskii equation. Therefore, in general, we can only hope to have phase retrieval up to an arbitrary constant.

\RED{In~\Cref{SEC:Nonuniqueness}, we present some simple examples to show that one can not expect uniqueness in phase retrieval with intensity data from a single medium, even if the medium is nonlinear and one can measure intensity data at infinitely many transverse planes. We will see, however, that data from multiple media can potentially allow us to restore this uniqueness.}

The rest of this work is organized as follows. In~\Cref{SEC:LinearMedia} we show that using data from multiple linear media, one is able to uniquely recover the phase of the incident wave up to a constant shift. We then show in~\Cref{SEC:NonlinearMedia} that this constant shift can be eliminated with data from multiple nonlinear media modeled by~\eqref{EQ:QNLS}. The theoretical derivations result in explicit phase reconstruction methods that we summarize in~\Cref{SEC:Algo}. Numerical experiments are supplied in~\Cref{SEC:Numer} to support the theoretical understanding. In~\Cref{SEC:Generalizations}, we investigate further the phase retrieval problem in settings where a mathematical understanding is missing, but the retrieval can still be done in a robust way. Concluding remarks are offered in~\Cref{SEC:Concl}.

\section{Recovering phase gradient with linear media diversity}
\label{SEC:LinearMedia}
 
We start by showing that appropriate intensity data measured from multiple linear media, that is, the case $\beta=0$, is sufficient to reconstruct the phase of the incident wave up to a constant shift.

Throughout the rest of the paper, we need to take various forms of derivatives of the intensity of the wave field $u$. Here, we make some general assumptions on the parameters of the problem so that one can reasonably believe that such differentiations are well-defined. We refer to~\cite{Cazenave-Book03,BeTa-JFA06,IwOg-TAMS15,KePoVe-TAMS96} and references therein for more detailed mathematical discussions on the regularity of solutions to the equations discussed here.\\[2ex]
\noindent{\bf Technical Assumptions:} We assume that
\begin{itemize}
    \item Either $\Omega=\bbR^d$ or $\Omega$ is a bounded set with smooth boundary $\partial\Omega$.
    \item The coefficients $\kappa$ and $\beta$ are sufficiently smooth, for instance, $\kappa, \beta \in\cP\subset \cC^2(\overline\Omega; \bbR)$ (the standard space of real-valued twice continuously differentiable functions on $\overline\Omega$), and $\kappa(\bx)\neq 0$ for all $\bx\in\Omega$.
    \item The incident wave field belongs to a sufficiently smooth class, for instance, $f\in\cF\subset \cC^2(\overline\Omega; \bbC)$, and $|f|(\bx)\neq 0$ for all $\bx\in \Omega$.
\end{itemize}
Such assumptions, together with appropriate boundary conditions when $\Omega$ is a bounded domain (for instance, homogeneous Dirichlet boundary conditions), are often sufficient to ensure the smoothness of solutions to~\eqref{EQ:LS},~\eqref{EQ:QNLS}, and~\eqref{EQ:GP} in $z$ and $\bx$ that we need in the calculations in this section and the next. Interested readers are referred to ~\cite{Cazenave-Book03} and references therein for rigorous mathematical treatments of the subject.

We have the following result on the uniqueness of reconstructing the gradient of the phase from intensity data measured in multiple media for model~\eqref{EQ:LS}.
\begin{theorem}\label{THM:k Unique} 
Let $f_1, f_2\in \cF$ be such that $\Lambda^{[0, \eps]}_{f_1}(\kappa,0) = \Lambda^{[0, \eps]}_{f_2}(\kappa,0)$ for all $\kappa\in \cP$ for the linear Schr\"odinger model~\eqref{EQ:LS}. Assume that there exists at least $N_m\ge d+2$ different $\kappa_1,\dots, \kappa_{N_m}\in \cP$ such that the matrix
    \begin{equation}\label{EQ:K}
    \cK(\bx):=
    \begin{pmatrix}
        \kappa_1 & \partial_{x_1} \kappa_1 & \cdots & \partial_{x_d} \kappa_1 & \Delta \kappa_1\\
        \vdots & \vdots & \ddots & \vdots &\vdots\\
        \kappa_m & \partial_{x_1} \kappa_m & \cdots & \partial_{x_d} \kappa_m & \Delta \kappa_m\\           
        \vdots & \vdots & \ddots & \vdots &\vdots\\
        \kappa_{N_m} & \partial_{x_1} \kappa_{N_m} & \cdots & \partial_{x_d} \kappa_{N_m} & \Delta \kappa_{N_m}
        \end{pmatrix}
    \end{equation}
    is full-rank for each $\bx\in\Omega$, and for each $\kappa_m$, $\left.\pdr{|u_1|}{z}\right|_{z=0} \neq 0$ and $\left.\pdr{|u_2|}{z}\right|_{z=0}\neq 0$ in $\Omega$. Here, $u_j$ ($j=1,2$) denotes the solution to \eqref{EQ:LS} with incident wave $f_j$ (which depends on $\kappa$). Then $f_2(\bx) = f_1(\bx) e^{i\varphi}$ for some constant $\varphi$.
\end{theorem}

To prove this uniqueness result, we utilize the following lemma.
\begin{lemma}\label{LEM:k Evolve} Let $u$ be a solution to \eqref{EQ:LS}. Decompose $u=|u|(v+iw)$ into its amplitude $|u|$ and direction $v+iw$, where $v$ and $w$ are real-valued functions. Then the time evolution equation \eqref{EQ:LS} for $u$ implies the following evolution equations for $|u|$, $v$, and $w$:
\begin{equation}\label{EQ:Ampl-Real-Imag}
    \dfrac{\partial |u|}{\partial z} = \kappa A(|u|,v,w),\qquad 
    \dfrac{\partial v}{\partial z} = \kappa B(|u|,v,w), \qquad 
    \dfrac{\partial w}{\partial z} = \kappa C(|u|,v,w)
\end{equation}
in $\Omega\times [0,\eps)$, where
\begin{equation}\label{EQ:ABC}
\begin{aligned}
    A(|u|, v, w) &:= |u|(w\Delta v - v\Delta w) + 2\nabla |u|\cdot (w\nabla v - v\nabla w),\\
    B(|u|, v, w) &:= -w\left[\frac{\Delta |u|}{|u|} + (w\Delta w + v\Delta v)\right],\\
    C(|u|, v, w) &:= v\left[\frac{\Delta |u|}{|u|} + (w\Delta w + v\Delta v)\right].
\end{aligned}
\end{equation}
\end{lemma}
\begin{proof} The proof is by direct calculation. Plugging in $u=|u|(v+iw)$ into the PDE \eqref{EQ:LS} and taking real and imaginary parts produces the system 
\begin{align*}
    v\dfrac{\partial |u|}{\partial z} + |u|\dfrac{\partial v}{\partial z} &= -\kappa\Big[(\Delta |u|)w + |u|\Delta w + 2\nabla |u|\cdot \nabla w\Big],\\
    w\dfrac{\partial |u|}{\partial z} + |u|\dfrac{\partial w}{\partial z} &= \kappa\Big[(\Delta |u|) v + |u|\Delta v + 2\nabla |u|\cdot \nabla v\Big],\\
    v\dfrac{\partial v}{\partial z} + w\dfrac{\partial w}{\partial z} &= 0\,.
\end{align*}
The last equation comes from differentiating the equation $v^2+w^2=1$ with respect to $z$. Solving this linear system for $\dfrac{\partial |u|}{\partial z}$,  $\dfrac{\partial v}{\partial z}$, and $\dfrac{\partial w}{\partial z}$ then gives the desired equations in~\eqref{EQ:Ampl-Real-Imag}.
\end{proof}

We also need the following additional lemma.
\begin{lemma}\label{LEM:Phase shift} Let $\nu_1, \omega_1, \nu_2, \omega_2\in \cC^1(\overline\Omega; \bbR)$ be functions on $\Omega$ such that $\nu_1^2+\omega_1^2 = \nu_2^2+\omega_2^2 = 1$ and 
\begin{equation}\label{EQ:grad phi}
\nu_1\nabla \omega_1 - \omega_1\nabla \nu_1 = \nu_2\nabla \omega_2 - \omega_2\nabla \nu_2
\end{equation}
in $\Omega$. Then there exists a constant $\varphi$ such that $\nu_1(\bx) + i \omega_1(\bx) = [\nu_2(\bx) + i \omega_2(\bx)]e^{i\varphi}$ for all $\bx\in\Omega$.
\end{lemma}

\begin{proof} It suffices to show that $\nabla \dfrac{\nu_1+i\omega_1}{\nu_2+i\omega_2} = 0$ in $\Omega$. We first compute 
\begin{equation*}
    \frac{\nu_1+i\omega_1}{\nu_2+i\omega_2} = (\nu_1\nu_2+\omega_1\omega_2) + i(\nu_2\omega_1 - \nu_1\omega_2)\,.
\end{equation*}
This immediately gives
\begin{equation}
\begin{aligned}
    \nabla \frac{\nu_1+i\omega_1}{\nu_2+i\omega_2} &= (\nu_1\nabla \nu_2 + \nu_2\nabla \nu_1 + \omega_1\nabla \omega_2 + \omega_2\nabla \omega_1)\\
    &\qquad + i(\nu_2\nabla \omega_1 + \omega_1\nabla \nu_2 - \nu_1\nabla \omega_2 - \omega_2\nabla \nu_1). 
\end{aligned}
\label{EQ:grad quotient}
\end{equation}
We know from taking the gradient of $\nu_j^2 + \omega_j^2= 1$ ($j=1,2$) that $\nu_j\nabla \nu_j + \omega_j\nabla \omega_j = 0$. It follows that there exists a vector-valued function $r_j(\bx)$ on $\Omega$ such that 
\begin{equation*}
    \nabla \nu_j = -\omega_j r_j,\qquad \nabla \omega_j = \nu_j r_j
\end{equation*}
in $\Omega$. Plugging this into \eqref{EQ:grad phi} then yields 
\begin{equation*}
    0 = \nu_1\nabla \omega_1 - \omega_1\nabla \nu_1 - \nu_2\nabla \omega_2 + \omega_2\nabla \nu_2 = (\nu_1^2 + \omega_1^2)r_1 - (\nu_2^2 + \omega_2^2)r_2 = r_1 - r_2,
\end{equation*}
so $r_1=r_2=:r$ and 
\begin{equation*}
    \nabla \nu_j = -\omega_j r,\qquad \nabla \omega_j = \nu_j r.
\end{equation*}
Finally, we can plug this into \eqref{EQ:grad quotient} and verify that both the real and imaginary parts are $0$ in $\Omega$:
\begin{equation*}
    \nu_1\nabla \nu_2 + \nu_2\nabla \nu_1 + \omega_1\nabla \omega_2 + \omega_2\nabla \omega_1 = (-\nu_1\omega_2 - \nu_2\omega_1 + \omega_1\nu_2 + \omega_2\nu_1)r = 0
\end{equation*}
and
\begin{align*}
    \nu_2\nabla \omega_1 + \omega_1\nabla \nu_2 - \nu_1\nabla \omega_2 - \omega_2\nabla \nu_1 = (\nu_2\nu_1 - \omega_1\omega_2 - \nu_1\nu_2 + \omega_2\omega_1)r = 0.
\end{align*}
The proof of the lemma is finished.
\end{proof}

\begin{proof}[Proof of Theorem~\ref{THM:k Unique}]
As in Lemma~\ref{LEM:k Evolve}, let us decompose $f_1$ and $f_2$ into their amplitudes and directions: 
\begin{align*}
    f_1 = a (\nu_1 + i\omega_1), \qquad f_2 = a (\nu_2 + i\omega_2).
\end{align*}
By Lemma~\ref{LEM:Phase shift}, it suffices to show that 
\begin{equation}
\nu_1\nabla \omega_1 - \omega_1\nabla \nu_1 = \nu_2\nabla \omega_2 - \omega_2\nabla \nu_2.
\end{equation}

For $j=1,2$, let $u_j$ denote the solution to \eqref{EQ:LS} with incident wave $f_j$, depending on $\kappa$. The evolution equation for the amplitude $|u|$ in~\eqref{EQ:Ampl-Real-Imag} says that for $j=1,2$,
\begin{align*}
    \dfrac{\partial |u_j|}{\partial z} = \kappa\Big[|u_j|(w_j\Delta v_j - v_j\Delta w_j) + 2\nabla |u_j|\cdot (w_j\nabla v_j - v_j\nabla w_j)\Big]
\end{align*}
At $z=0$, this equation reduces to 
\begin{align}
    \left.\dfrac{\partial |u_j|}{\partial z}\right|_{z=0} = \kappa\Big[a(\omega_j\Delta \nu_j - \nu_j\Delta \omega_j) + 2\nabla a\cdot (\omega_j\nabla \nu_j - \nu_j\nabla \omega_j)\Big].
\end{align}
Thanks to the assumptions that $|u_1|=|u_2|$ in $[0,\eps]\times\Omega$ and $\kappa\neq 0$ in $\Omega$, we realize that 
\begin{equation}\label{EQ:A}
    a(\omega_j\Delta \nu_j - \nu_j\Delta \omega_j) + 2\nabla a\cdot (\omega_j\nabla \nu_j - \nu_j\nabla \omega_j)
\end{equation}
is independent of $j\in \{1,2\}$. 

Next, we repeat the above procedure by first differentiating the evolution equation for the amplitude with respect to $z$ to obtain, for $j=1,2$,
\begin{align*}
    \frac{1}{\kappa}\dfrac{\partial^2 |u_j|}{\partial z^2} &= \kappa A(|u_j|,v_j,w_j)(w_j\Delta v_j - v_j\Delta w_j)\\
    &+ \kappa |u_j|C(|u_j|,v_j,w_j)\Delta v_j - \kappa |u_j| B(|u_j|,v_j,w_j)\Delta w_j\\
    &+ |u_j|w_j\Delta\Big[\kappa B(|u_j|,v_j,w_j)\Big] - |u_j|v_j\Delta \Big[\kappa C(|u_j|,v_j,w_j)\Big]\\
    &+2\nabla\Big[\kappa A(|u_j|,v_j,w_j)\Big]\cdot (w_j\nabla v_j - v_j\nabla w_j)\\
    &+ 2\Big[\kappa C(|u_j|,v_j,w_j)\Big]\nabla |u_j|\cdot\nabla v_j - 2\Big[\kappa B(|u_j|,v_j,w_j)\Big]\nabla |u_j|\cdot\nabla w_j \\
    &+ 2w_j\nabla |u_j|\cdot\nabla\Big[\kappa B(|u_j|,v_j,w_j)\Big] - 2v_j\nabla |u_j|\cdot\nabla\Big[\kappa C(|u_j|,v_j,w_j)\Big].
\end{align*}
At $z=0$, this equation reduces to
\begin{align*}
    \frac{1}{\kappa}\left.\dfrac{\partial^2 |u_j|}{\partial z^2}\right|_{z=0} &= \kappa A(a,\nu_j,\omega_j)(\omega_j\Delta \nu_j - \nu_j\Delta \omega_j)\\
    &+ \kappa a C(a,\nu_j,\omega_j)\Delta \nu_j - \kappa a B(a,\nu_j,\omega_j)\Delta \omega_j\\
    &+ a\omega_j\Delta\Big[\kappa B(a,\nu_j,\omega_j)\Big] - a\nu_j\Delta \Big[\kappa C(a,\nu_j,\omega_j)\Big]\\
    &+2\nabla\Big[\kappa A(a,\nu_j,\omega_j)\Big]\cdot (\omega_j\nabla \nu_j - \nu_j\nabla \omega_j)\\
    &+ 2\Big[\kappa C(a,\nu_j,\omega_j)\Big]\nabla a\cdot\nabla \nu_j - 2\Big[\kappa B(a,\nu_j,\omega_j)\Big]\nabla a\cdot\nabla \omega_j \\
    &+ 2\omega_j\nabla a\cdot\nabla\Big[\kappa B(a,\nu_j,\omega_j)\Big] - 2\nu_j\nabla a\cdot\nabla\Big[\kappa C(a,\nu_j,\omega_j)\Big].
\end{align*}
To make sense of the complicated right-hand side, we recognize that after multiplying through by $\kappa$, it is of the form 
\begin{align}\label{EQ:F}
F_1(a,\nu_j,\omega_j)\kappa^2 + F_2(a,\nu_j,\omega_j)\cdot \kappa\nabla \kappa + F_3(a,\nu_j,\omega_j)\kappa\Delta\kappa,
\end{align}
where $F_1$, $F_2$, and $F_3$ each depend only on $a$, $\nu_j$, $\omega_j$ and their spatial derivatives.

Thanks to the assumption that $|u_1|=|u_2|$ in $[0,\eps]\times\Omega$, we realize that the right-hand side is the same for $j=1,2$. Hence
\begin{align*}
    &\Big[F_1(a,\nu_1,\omega_1)-F_1(a,\nu_2,\omega_2)\Big]\kappa^2 + \Big[F_2(a,\nu_1,\omega_1)-F_2(a,\nu_2,\omega_2)\Big]\cdot \kappa\nabla\kappa\\
    &\qquad + \Big[F_3(a,\nu_1,\omega_1)-F_3(a,\nu_2,\omega_2)\Big]\kappa\Delta\kappa = 0.
\end{align*}
Recall that this holds for all $\kappa\in\cP$. By the assumption that the matrix $\cK$ in~\eqref{EQ:K} is invertible, we conclude that each of the three coefficients must be $0$.

In particular, from the third coefficient we have $F_3(a,\nu_1,\omega_1) = F_3(a,\nu_2,\omega_2)$. An inspection of the definition of $F_3$ reveals that
\begin{equation}\label{EQ:F3 k}
    F_3(a,\nu_j,\omega_j) = a\omega_j B(a,\nu_j,\omega_j) - a\nu_j C(a,\nu_j,\omega_j).
\end{equation}
Therefore, after dividing through by $a$ (which is nonzero in $\Omega$ by assumption) we obtain that 
\begin{equation*}
    \omega_j B(a,\nu_j,\omega_j) - \nu_j C(a,\nu_j,\omega_j)
\end{equation*}
is independent of $j\in\{1,2\}$. In the following, it will be useful to note also that its gradient 
\begin{align*}
    \nabla\Big[\omega_j B(a,\nu_j,\omega_j) - \nu_j C(a,\nu_j,\omega_j)\Big] &= (\omega_j\nabla B - \nu_j\nabla C) + (\nabla\omega_j B - \nabla\nu_j C) = \omega_j\nabla B - \nu_j\nabla C.
\end{align*}
is independent of $i\in\{1,2\}$ as well. (For clarity, we have suppressed the dependence of $B$ and $C$ on $a$, $\nu_j$, $\omega_j$.) To justify the last step above, recall the definitions \eqref{EQ:ABC} for $B$ and $C$ and observe that 
\begin{align*}
    \nabla\omega_j B - \nabla\nu_j C &= -(\omega_j\nabla\omega_j + \nu_j\nabla\nu_j)\left[\frac{\Delta a}{a} + (\omega_j\Delta \omega_j + \nu_j\Delta \nu_j)\right]\\
    &= 0\cdot\left[\frac{\Delta a}{a} + (\omega_j\Delta \omega_j + \nu_j\Delta \nu_j)\right] = 0,
\end{align*}
where $\omega_j\nabla\omega_j + \nu_j\nabla\nu_j = 0$ comes from taking the gradient of the equation $\omega_j^2+\nu_j^2 = 1$.

Likewise, from the second coefficient we have $F_2(a,\nu_1,\omega_1) = F_2(a,\nu_2,\omega_2)$. An inspection of the definition of $F_2$ reveals that 
\begin{equation}\label{EQ:F2}
\begin{aligned}
    F_2(a,\nu_j,\omega_j) &= 2a\omega_j \nabla B - 2a\nu_j\nabla C + 2A(\omega_j\nabla\nu_j - \nu_j\nabla\omega_j) + 2\omega_j B\nabla a - 2\nu_j C\nabla a\\
    &= 2A(\omega_j\nabla\nu_j - \nu_j\nabla\omega_j) + 2a(\omega_j\nabla B - \nu_j\nabla C) + 2(\omega_j B - \nu_j C)\nabla a.
\end{aligned}
\end{equation}
We just established in the previous paragraph that the second and third terms are independent of $j\in\{1,2\}$. Since the left-hand side is also independent of $j$, we deduce that the first term on the right-hand side 
\begin{equation*}
    2A(a,\nu_j,\omega_j)(\omega_j\nabla\nu_j - \nu_j\nabla\omega_j)    
\end{equation*}
is independent of $j$. Now recall from the first part of the proof that $A(a,\nu_j,\omega_j)$ is independent of $j$. Furthermore, the function
\begin{equation*}
    A(a,\nu_j,\omega_j) = \frac{1}{\kappa}\left.\pdr{|u_j|}{z}\right|_{z=0}
\end{equation*}
is nonzero in $\Omega$ by the assumption in the theorem statement. It follows that 
\begin{equation*}
    \omega_j\nabla\nu_j - \nu_j\nabla\omega_j
\end{equation*}
is independent of $j\in\{1,2\}$, as desired. The proof is complete.
\end{proof}

The same calculations can be done for the Gross-Pitaevskii~\eqref{EQ:GP} equation.
\begin{corollary}\label{COR:GP Unique} 
Fix $\beta\in\cP$. Let $f_1, f_2\in \cF$ be such that $\Lambda^{[0, \eps]}_{f_1}(\kappa,\beta) = \Lambda^{[0, \eps]}_{f_2}(\kappa,\beta)$ for all $\kappa\in\cP$ for the Gross-Pitaevskii model~\eqref{EQ:GP}. Then, under the same assumptions of~\Cref{THM:k Unique}, we have that $f_2(\bx) = f_1(\bx) e^{i\varphi}$ for some constant $\varphi$.
\end{corollary}
\begin{proof}
The proof follows exactly the same procedure as that of~\Cref{THM:k Unique}. We only highlight the main differences. Let $u$ be a solution to \eqref{EQ:GP} and take the decomposition $u=|u|(v+iw)$. Then
the evolution equations for $|u|$, $v$, and $w$ are
\begin{equation}
\dfrac{\partial |u|}{\partial z} =\kappa A(|u|, v, w), \quad \dfrac{\partial v}{\partial z} =\kappa B(|u|, v, w)-\beta |u|^2 w, \qquad \dfrac{\partial w}{\partial z} =\kappa C(|u|, v, w)+\beta |u|^2 v,
\end{equation}
where the functions $A$, $B$, and $C$ are defined in~\eqref{EQ:ABC}.

The second-order derivative of $|u|$ at $z=0$ in this case can be written as
\begin{align*}
    \frac{1}{\kappa}\left.\dfrac{\partial^2 |u_j|}{\partial z^2}\right|_{z=0} &= \kappa A(a,\nu_j,\omega_j)(\omega_j\Delta \nu_j - \nu_j\Delta \omega_j)\\
    & + \kappa a C(a,\nu_j,\omega_j)\Delta \nu_j - \kappa a B(a,\nu_j,\omega_j)\Delta \omega_j\\
    & + a\omega_j\Delta\Big[\kappa B(a,\nu_j,\omega_j)\Big] - a\nu_j\Delta \Big[\kappa C(a,\nu_j,\omega_j)\Big]\\
    & +2\nabla\Big[\kappa A(a,\nu_j,\omega_j)\Big]\cdot (\omega_j\nabla \nu_j - \nu_j\nabla \omega_j)\\
    & + 2\Big[\kappa C(a,\nu_j,\omega_j)\Big]\nabla a\cdot\nabla \nu_j - 2\Big[\kappa B(a,\nu_j,\omega_j)\Big]\nabla a\cdot\nabla \omega_j \\
    & + 2\omega_j\nabla a\cdot\nabla\Big[\kappa B(a,\nu_j,\omega_j)\Big] - 2\nu_j\nabla a\cdot\nabla\Big[\kappa C(a,\nu_j,\omega_j)\Big]\\
    & + a\Big[\beta a^2 \nu_j\Delta \nu_j -\omega_j\Delta (\beta a^2 \omega_j)+(\beta a^2 \omega_j)\Delta\omega_j-\nu_j\Delta (\beta a^2 \nu_j)\Big]\\
    &+2\nabla a \cdot \Big[\beta a^2 \nu_j\nabla \nu_j-\omega_j\nabla(\beta a^2 \omega_j)+(\beta a^2 \omega_j)\nabla \omega_j-\nu_j\nabla(\beta a^2 \nu_j)\Big].
\end{align*}
It turns out that the last two terms simplify to
\[
    F_0(a, \beta):=-a\Delta(\beta a^2)-\nabla a\cdot \nabla (\beta a^2)
\]
which is independent of $\nu_j$ and $\omega_j$. Therefore, the right-hand side, after multiplying through by $\kappa$, is of the form 
\begin{align}
F_0(a, \beta) \kappa+ F_1(a,\nu_j,\omega_j)\kappa^2 + F_2(a,\nu_j,\omega_j)\cdot \kappa\nabla \kappa + F_3(a,\nu_j,\omega_j)\kappa\Delta\kappa,
\end{align}
where $F_1$, $F_2$, and $F_3$ are the same as defined in~\eqref{EQ:F}.

Thanks to the assumption that $|u_1|=|u_2|$ in $[0,\eps]\times\Omega$, we realize that the right-hand side is the same for $j=1,2$. Hence
\begin{align*}
    &\Big[F_1(a,\nu_1,\omega_1)-F_1(a,\nu_2,\omega_2)\Big]\kappa^2 + \Big[F_2(a,\nu_1,\omega_1)-F_2(a,\nu_2,\omega_2)\Big]\cdot \kappa\nabla\kappa\\
    &\qquad + \Big[F_3(a,\nu_1,\omega_1)-F_3(a,\nu_2,\omega_2)\Big]\kappa\Delta\kappa = 0.
\end{align*}
The rest follows the same argument as that in~\Cref{THM:k Unique}.
\end{proof}

\begin{remark}\label{REMARK:beta in GP}
It is clear from the proof of~\Cref{COR:GP Unique} that one can not use $\beta$ to help with phase reconstruction, as $\beta$ is not coupled to the phase variables $\nu$ and $\omega$ in $F_0$. This contrasts with the result we will show in the next section on the nonlinearity described by the quadratic nonlinear Schr\"{o}dinger model~\eqref{EQ:QNLS}.
\end{remark}

\section{Recovering phase with nonlinear media diversity}
\label{SEC:NonlinearMedia}
    
We now show a case where the nonlinearity of the media improves phase retrieval. More precisely, for the quadratic nonlinear Schr\"{o}dinger model~\eqref{EQ:QNLS}, we show that one can uniquely reconstruct the phase of the incident wave. That is, nonlinearity allows us to eliminate the arbitrary constant phase shift in the reconstruction.

To focus only on the nonlinearity, we fix the coefficient $\kappa=1$.  

\begin{theorem}\label{THM:Unique} 
Let $f_1, f_2\in \cF$ be such that $\Lambda^{[0, \eps]}_{f_1}(1,\beta) = \Lambda^{[0, \eps]}_{f_2}(1,\beta)$ for all $\beta\in\cP$ for the quadratic nonlinear Schr\"{o}dinger model~\eqref{EQ:QNLS}. Assume that there exists at least $N_m\ge d+4$ different $\beta_1,\dots,\beta_{N_m}\in \cP$ such that the matrix
    \begin{equation}\label{EQ:Beta}
    \cB(\bx):=
    \begin{pmatrix}
        1& \beta_1 & \beta_1^2 & \partial_{x_1} \beta_1 & \cdots & \partial_{x_d} \beta_1 & \Delta \beta_1\\
        \vdots & \vdots & \vdots & \vdots & \ddots & \vdots &\vdots\\
        1& \beta_m & \beta_m^2 & \partial_{x_1} \beta_m & \cdots & \partial_{x_d} \beta_m & \Delta \beta_m\\           
        \vdots & \vdots & \vdots & \vdots & \ddots & \vdots &\vdots\\
        1&\beta_{N_m} & \beta_{N_m}^2 & \partial_{x_1} \beta_{N_m} & \cdots & \partial_{x_d} \beta_{N_m} & \Delta \beta_{N_m}
        \end{pmatrix}
    \end{equation}
    is full-rank for each $\bx\in\Omega$. Then $f_1=f_2$.
\end{theorem}

To prove this uniqueness result, we utilize the following lemma. Its proof follows the same line of calculations as the proof of Lemma \ref{LEM:k Evolve}, so we omit it. 
\begin{lemma}\label{LEM:Evolve} Let $u$ be a solution to \eqref{EQ:QNLS} with $\kappa = 1$ and decompose $u=|u|(v+iw)$ into its amplitude $|u|$ and direction $v+iw$, where $v$ and $w$ are real-valued functions. Then the time evolution equation \eqref{EQ:QNLS} for $u$ implies the following evolution equations for $|u|$, $v$, and $w$:
\begin{equation}\label{EQ:Ampl-Real-Imag QNLS}
\begin{aligned}
    \dfrac{\partial |u|}{\partial z} &= A(|u|,v,w)-\beta|u|^2 w,\\ \qquad 
    \dfrac{\partial v}{\partial z} &= B(|u|,v,w)-\beta|u|vw,\\ \qquad 
    \dfrac{\partial w}{\partial z} &= C(|u|,v,w)+\beta|u|v^2\,.
\end{aligned}
\end{equation}
in $\Omega\times [0,\eps)$, where the functions $A$, $B$, and $C$ are defined as in~\eqref{EQ:ABC}.
\end{lemma}

\begin{proof}[Proof of Theorem~\ref{THM:Unique}]
Following the same procedure as before, let us decompose $f_1$ and $f_2$ into their amplitudes and directions: 
\begin{align*}
    f_1 = a (\nu_1 + i\omega_1), \qquad f_2 = a (\nu_2 + i\omega_2).
\end{align*}
We want to show that $f_1=f_2$, i.e. $\nu_1 = \nu_2$ and $\omega_1 = \omega_2$.

Let $\beta\in \cP$, and for $j=1,2$ let $u_j$ denote the solution to \eqref{EQ:QNLS} with initial condition $f_j$. The evolution equation for the amplitude in~\eqref{EQ:Ampl-Real-Imag QNLS}  says that for $j=1,2$,
\begin{align*}
    \dfrac{\partial |u_j|}{\partial z} = |u_j|(w_j\Delta v_j - v_j\Delta w_j) + 2\nabla |u_j|\cdot (w_j\nabla v_j - v_j\nabla w_j) - \beta |u_j|^2 w_j
\end{align*}
At $z=0$, this equation reduces to 
\begin{align}\label{EQ:retrieve omega}
    \left.\dfrac{\partial |u_j|}{\partial z}\right|_{z=0} = a(\omega_j\Delta \nu_j - \nu_j\Delta \omega_j) + 2\nabla a\cdot (\omega_j\nabla \nu_j - \nu_j\nabla \omega_j) - \beta a^2 \omega_j.
\end{align}
Thanks to the assumption that $|u_1|=|u_2|$ in $[0,\eps]\times\Omega$, we realize that 
\begin{align*}
    a(\omega_1\Delta \nu_1 - \nu_1\Delta \omega_1) + 2\nabla a\cdot (\omega_1\nabla \nu_1& - \nu_1\nabla \omega_1) - \beta a^2 \omega_1 = \\
    &a(\omega_2\Delta \nu_2 - \nu_2\Delta \omega_2) + 2\nabla a\cdot (\omega_2\nabla \nu_2 - \nu_2\nabla \omega_2) - \beta a^2 \omega_2,
\end{align*}
or after rearranging,
\begin{align*}
    \Big[a(\omega_1\Delta \nu_1 - \nu_1\Delta \omega_1) + 2\nabla a\cdot (\omega_1\nabla \nu_1& - \nu_1\nabla \omega_1) - a(\omega_2\Delta \nu_2 - \nu_2\Delta \omega_2) - 2\nabla a\cdot (\omega_2\nabla \nu_2 - \nu_2\nabla \omega_2)\Big]\\
    &- \beta a^2(\omega_1 - \omega_2) = 0.
\end{align*}
Recalling that this holds for all $\beta\in \cP$ allows us to pick two different $\beta$ (for example taking $\beta\equiv 1$ and $\beta\equiv 2$) to conclude that 
\begin{align*}
    a(\omega_1\Delta \nu_1 - \nu_1\Delta \omega_1) + 2\nabla a\cdot (\omega_1\nabla \nu_1 - \nu_1\nabla \omega_1) - a(\omega_2\Delta \nu_2 - \nu_2\Delta \omega_2) \qquad \qquad & \\  - 2\nabla a\cdot (\omega_2\nabla \nu_2 - \nu_2\nabla \omega_2) &= 0,\\[1ex]
    -a^2(\omega_1 - \omega_2) &= 0.
\end{align*}
The second equation, in particular, informs us that 
\begin{align}\label{EQ:omega}
    \omega_1 = \omega_2.
\end{align}

It remains to show that $\nu_1 = \nu_2$. To do this, we repeat the above procedure by first differentiating the amplitude evolution equation in~\eqref{EQ:Ampl-Real-Imag QNLS} with respect to $z$ to obtain, for $j=1,2$,
\begin{align}\label{EQ:Evolve Amp2}
    \dfrac{\partial^2 |u_j|}{\partial z^2} &= \dfrac{\partial |u_j|}{\partial z}(w_j\Delta v_j - v_j\Delta w_j) + |u_j|\left[\dfrac{\partial w_j}{\partial z}\Delta v_j - \dfrac{\partial v_j}{\partial z}\Delta w_j + w_j\Delta \dfrac{\partial v_j}{\partial z} - v_j\Delta \dfrac{\partial w_j}{\partial z}\right]\nonumber\\
    &\qquad +2\nabla \dfrac{\partial |u_j|}{\partial z}\cdot (w_j\nabla v_j - v_j\nabla w_j) \\
    &\qquad + 2\nabla |u_j|\cdot \left[\dfrac{\partial w_j}{\partial z}\nabla v_j - \dfrac{\partial v_j}{\partial z}\nabla w_j + w_j\nabla \dfrac{\partial v_j}{\partial z} - v_j\nabla \dfrac{\partial w_j}{\partial z}\right] \nonumber\\
    &\qquad -\beta\Big[2|u_j|\dfrac{\partial |u_j|}{\partial z}w_j + |u_j|^2\dfrac{\partial w_j}{\partial z}\Big].\nonumber
\end{align}
Plugging the equations in~\eqref{EQ:Ampl-Real-Imag QNLS} into \eqref{EQ:Evolve Amp2} yields 
\begin{align*}
    &\dfrac{\partial^2 |u_j|}{\partial z^2} = \Big[A(|u_j|,v_j,w_j) - \beta |u_j|^2 w_j\Big](w_j\Delta v_j - v_j\Delta w_j)\\
    &+ |u_j|\Big[C(|u_j|,v_j,w_j) + \beta |u_j| (v_j)^2\Big]\Delta v_j - |u_j| \Big[B(|u_j|,v_j,w_j) - \beta |u_j| v_j w_j\Big]\Delta w_j\\
    &+ |u_j|w_j\Delta\Big[B(|u_j|,v_j,w_j) - \beta |u_j| v_j w_j\Big] - |u_j|v_j\Delta \Big[C(|u_j|,v_j,w_j) + \beta |u_j| (v_j)^2\Big]\\
    &+2\nabla\Big[A(|u_j|,v_j,w_j) - \beta |u_j|^2 w_j\Big]\cdot (w_j\nabla v_j - v_j\nabla w_j)\\
    &+ 2\Big[C(|u_j|,v_j,w_j) + \beta |u_j| (v_j)^2\Big]\nabla |u_j|\cdot\nabla v_j - 2\Big[B(|u_j|,v_j,w_j) - \beta |u_j| v_j w_j\Big]\nabla |u_j|\cdot\nabla w_j \\
    & + 2w_j\nabla |u_j|\cdot\nabla\Big[B(|u_j|,v_j,w_j) - \beta |u_j| v_j w_j\Big] - 2v_j\nabla |u_j|\cdot\nabla\Big[C(|u_j|,v_j,w_j) + \beta |u_j| (v_j)^2\Big]\\
    & -2\beta |u_j|\Big[A(|u_j|,v_j,w_j) - \beta |u_j|^2 w_j\Big]w_j -\beta|u_j|^2\Big[C(|u_j|,v_j,w_j) + \beta |u_j| (v_j)^2\Big]\,.
\end{align*}
At $z=0$, this equation reduces to
\begin{align*}
    \left.\dfrac{\partial^2 |u_j|}{\partial z^2}\right|_{z=0} &= \Big[A(a,\nu_j,\omega_j) - \beta a^2 \omega_j\Big](\omega_j\Delta \nu_j - \nu_j\Delta \omega_j)\\
    &+ a\Big[C(a,\nu_j,\omega_j) + \beta a \nu_j^2\Big]\Delta \nu_j - a \Big[B(a,\nu_j,\omega_j) - \beta a \nu_j \omega_j\Big]\Delta \omega_j\\
    &  + a\omega_j\Delta\Big[B(a,\nu_j,\omega_j) - \beta a \nu_j \omega_j\Big] - a\nu_j\Delta \Big[C(a,\nu_j,\omega_j) + \beta a \nu_j^2\Big]\\
    & +2\nabla\Big[A(a,\nu_j,\omega_j) - \beta a^2 \omega_j\Big]\cdot (\omega_j\nabla \nu_j - \nu_j\nabla \omega_j)\\
    &  + 2\Big[C(a,\nu_j,\omega_j) + \beta a \nu_j^2\Big]\nabla a\cdot\nabla \nu_j - 2\Big[B(a,\nu_j,\omega_j) - \beta a \nu_j \omega_j\Big]\nabla a\cdot\nabla \omega_j \\
    & + 2\omega_j\nabla a\cdot\nabla\Big[B(a,\nu_j,\omega_j) - \beta a \nu_j \omega_j\Big] - 2\nu_j\nabla a\cdot\nabla\Big[C(a,\nu_j,\omega_j) + \beta a \nu_j^2\Big]\\
    & -2\beta a\Big[A(a,\nu_j,\omega_j) - \beta a^2 \omega_j\Big]\omega_j -\beta a^2\Big[C(a,\nu_j,\omega_j) + \beta a \nu_j^2\Big]\,.
\end{align*}
To make sense of the complicated right-hand side, we recognize that it is of the form 
\begin{align}
    F_1(a,\nu_j,\omega_j) + F_2(a,\nu_j,\omega_j)\beta + F_3(a,\nu_j,\omega_j)\beta^2 + F_4(a,\nu_j,\omega_j)\cdot \nabla\beta + F_5(a,\nu_j,\omega_j)\Delta\beta,
    \label{EQ:retrieve nu 1}
\end{align}
where $F_1$, $F_2$, $F_3$, $F_4$, and $F_5$ each depend only on $a$, $\nu_j$, $\omega_j$ and their spatial derivatives.

Thanks to the assumption that $|u_1|=|u_2|$ in $[0,\eps]\times\Omega$, we realize that the right-hand side is the same for $j=1,2$. Hence
\begin{align*}
    &\Big[F_1(a,\nu_1,\omega_1)-F_1(a,\nu_2,\omega_2)\Big] + \Big[F_2(a,\nu_1,\omega_1)-F_2(a,\nu_2,\omega_2)\Big]\beta + \Big[F_3(a,\nu_1,\omega_1)-F_3(a,\nu_2,\omega_2)\Big]\beta^2\nonumber\\
    &\qquad + \Big[F_4(a,\nu_1,\omega_1)-F_4(a,\nu_2,\omega_2)\Big]\cdot \nabla\beta + \Big[F_5(a,\nu_1,\omega_1)-F_5(a,\nu_2,\omega_2)\Big]\Delta\beta = 0. 
\end{align*}
By the assumption that the $\beta$ matrix $\cB$ in~\eqref{EQ:Beta} is invertible, we conclude that 
\[
    F_k(a, \nu_1, \omega_1)=F_k(a, \nu_2, \omega_2), \quad 1\le k\le 5\,.
\]

In particular, from the fifth coefficient we have $F_5(a,\nu_1,\omega_1) = F_5(a,\nu_2,\omega_2)$. An inspection of the definition of $F_5$ reveals that
\begin{equation}\label{EQ:retrieve nu 2}
    F_5(a,\nu_j,\omega_j) = -(a\omega_j)( a\nu_j\omega_j) - (a\nu_j)(a\nu_j^2) = -a^2\nu_j(\nu_j^2+\omega_j^2) = -a^2\nu_j.
\end{equation}
Thus $-a^2\nu_1 = -a^2\nu_2$, and so
\begin{align}\label{EQ:nu}
   \nu_1 = \nu_2.
\end{align}
The proof is complete upon combining this equality with \eqref{EQ:omega}.
\end{proof}

\begin{remark}
    The same calculation can be conducted to show that this technique also works for nonlinear media with the term $\beta u^2$ replaced with $\beta u^3$. However, the technique also fails to work for linear media with the term $\beta u$. Also, the technique will not work with media governed by the nonlinearity given by the Gross-Pitaevskii model~\eqref{EQ:GP}; see~\Cref{REMARK:beta in GP}.
\end{remark}

\section{Analytic phase reconstruction}
\label{SEC:Algo}

One key feature of the uniqueness proofs in the previous sections is that the theoretical derivations can be directly implemented as phase reconstruction algorithms.

\subsection{Reconstruction from \texorpdfstring{$\kappa$}{} diversity}

From $\kappa$ diversity, we can reconstruct the phase of $f=a(\nu+i\omega)=a e^{i\varphi}$ up to a constant. That is, we can reconstruct the phase gradient
\[
    -\nabla \varphi =\omega\nabla \nu-\nu \nabla \omega\,.
\]
From the proof of~\Cref{THM:k Unique}, more precisely~\eqref{EQ:F}, we have that
\begin{align*}
    \frac{1}{\kappa}\left.\dfrac{\partial^2 |u|}{\partial z^2}\right|_{z=0}=F_1(a,\nu,\omega)\kappa + F_2(a,\nu,\omega)\cdot \nabla \kappa + F_3(a,\nu,\omega) \Delta\kappa
\end{align*}
This allows us to form the linear system
\begin{equation}\label{EQ:K System}
        \cK 
        \begin{pmatrix}
        F_1\\ F_{2,x_1}\\ \vdots \\ \vdots\\
            \vdots\\ F_{2, x_d} \\ F_3    
        \end{pmatrix}
        =\begin{pmatrix}
            \frac{1}{\kappa_1}\frac{\partial^2 |u_1|}{\partial z^2}|_{z=0}\\
            \frac{1}{\kappa_2}\frac{\partial^2 |u_2|}{\partial z^2}|_{z=0}\\
            \vdots\\
            \vdots\\
            \vdots\\
            \frac{1}
            {\kappa_{N_m}}\frac{\partial^2 |u_{N_m}|}{\partial z^2}|_{z=0}            
        \end{pmatrix}
\end{equation}
at every $\bx\in\Omega$. For $m=1,2,\dots, N_m$, the notation $u_m$ denotes the solution to \eqref{EQ:LS} with $\kappa = \kappa_m$. Given measured intensity data, we can evaluate the right-hand side by a finite difference scheme, and then solve the system, in the least-squares sense when $N_m\ge d+2$, to find $F_1$, $F_2$, and $F_3$. Meanwhile, from~\eqref{EQ:F3 k}, we have that 
\begin{equation*}
    \frac{1}{a} F_3(a,\nu,\omega) = \omega B(a,\nu,\omega) - \nu C(a,\nu,\omega).
\end{equation*}
This, together with ~\eqref{EQ:F2}, leads to the reconstruction of the phase gradient as
\begin{equation}\label{EQ:Grad-Phi-Rec}
    \omega\nabla\nu - \nu\nabla\omega=\frac{1}{2A}F_2(a,\nu,\omega) -\frac{1}{A}\nabla F_3\,.
\end{equation}
Therefore, from the data, we compute $A=\dfrac{1}{\kappa}\dfrac{\partial|u|}{\partial z}\Big|_{z=0}$ using the first equation in~\eqref{EQ:Ampl-Real-Imag}. We then form the $\cK$ matrix and solve for $F_1$, $F_2$, and $F_3$. The final reconstruction follows from~\eqref{EQ:Grad-Phi-Rec}. 

In practice, we can estimate $A$ as the average of its value from all $N_m$ media as
\begin{equation}\label{EQ:Average of A}
    A=\frac{1}{N_m}\sum_{m=1}^{N_m}\dfrac{1}{\kappa_m}\left.\dfrac{\partial|u_m|}{\partial z}\right|_{z=0}\,.
\end{equation}
This averaging is useful to average out noise in measured data.

\begin{algorithm}[!htb]
\caption{Reconstructing Phase Gradient $\nabla \varphi$}
\label{ALGO:Phase Grad Rec}
\, Input: $a:=|f|$, $\{\kappa_m, |u_m|\}_{m=1}^{N_m}$
\begin{algorithmic}[1]
  \State Evaluate $\{\frac{1}
            {\kappa_{m}}\frac{\partial |u_m|}{\partial z}|_{z=0}, \frac{1}
            {\kappa_{m}}\frac{\partial^2 |u_m|}{\partial z^2}|_{z=0}\}_{m=1}^{N_m}$, and compute $A$ from~\eqref{EQ:Average of A}
  \State Form the matrix $\cK$ in~\eqref{EQ:K} and the right-hand side of~\eqref{EQ:K System}
  \State Solve~\eqref{EQ:K System} for $\{F_1, F_2, F_3\}$
  \State Reconstruct $-\nabla \varphi$ following~\eqref{EQ:Grad-Phi-Rec}
\end{algorithmic}
\end{algorithm}
 
The complete procedure is prescribed in~\Cref{ALGO:Phase Grad Rec}. In the algorithm, we used finite difference approximation to the first- and second-order derivatives of $|u|$ with respect to $z$.

\subsection{Reconstruction from \texorpdfstring{$\beta$}{} diversity}

With $\beta$ diversity, we can reconstruct the phase of the incident wave $f(\bx)=a(\bx)\big(\nu(\bx)+i\omega(\bx)\big)$ uniquely. This is done by reconstructing the real and imaginary parts of $f$.

From~\eqref{EQ:retrieve omega}, we have that the data corresponding to medium $\beta_m$ ($1\le m\le N_m$) at $z=0$ satisfy
\begin{equation*}
    \left.\dfrac{\partial |u_m|}{\partial z}\right|_{z=0} = a(\omega\Delta \nu - \nu\Delta \omega) + 2\nabla a\cdot (\omega\nabla \nu - \nu\nabla \omega) - \beta_m a^2 \omega\,.
\end{equation*}
Taking two different $\beta$ gives us
\[
 \left.\dfrac{\partial |u_1|}{\partial z}\right|_{z=0} - \left.\dfrac{\partial |u_2|}{\partial z}\right|_{z=0} = a^2(\beta_2 - \beta_1) \omega,
\]
which leads to the reconstruction of $\omega$ as
\begin{equation}\label{EQ:omega Rec}
 \omega=\frac{1}{a^2(\beta_2 - \beta_1)}\left(\left.\dfrac{\partial |u_1|}{\partial z}\right|_{z=0} - \left.\dfrac{\partial |u_2|}{\partial z}\right|_{z=0}\right)\,.
\end{equation}

Meanwhile, from~\eqref{EQ:retrieve nu 1}, we have
\[
\left.\frac{\partial^2 |u_m|}{\partial z^2}\right|_{z=0} = F_1(a,\nu,\omega) + F_2(a,\nu,\omega)\beta_m + F_3(a,\nu,\omega)\beta_m^2 + F_4(a,\nu,\omega)\cdot \nabla\beta_m + F_5(a,\nu,\omega)\Delta\beta_m\,.
\]
We can therefore form the linear system
\begin{equation}\label{EQ:B System}
        \cB 
        \begin{pmatrix}
        F_1\\ F_2\\ F_3\\ F_{4,x_1}\\ \vdots \\ F_{4, x_d} \\ F_5
        \end{pmatrix}
        =\begin{pmatrix}
             \frac{\partial^2 |u_1|}{\partial z^2}|_{z=0}\\
             \frac{\partial^2 |u_2|}{\partial z^2}|_{z=0}\\
            \vdots\\
            \vdots\\
            \vdots\\
             \frac{\partial^2 |u_{N_m}|}{\partial z^2}|_{z=0}   
        \end{pmatrix}
\end{equation}
at every $\bx\in\Omega$. We can invert this system to compute $\{F_1, \cdots, F_5\}$. According to ~\eqref{EQ:retrieve nu 2}, $\nu$ can be reconstructed as
\begin{equation}\label{EQ:nu Rec}
\nu(\bx) =-\frac{1}{a^2} F_5(\bx)\,.
\end{equation}
The reconstruction procedure is summarized in~\Cref{ALGO:Phase Rec}.
\begin{algorithm}[!htb]
\caption{Reconstructing Phase $\varphi$}
\label{ALGO:Phase Rec}
\, Input: $a:=|f|$, $\{\beta_m, |u_m|\}_{m=1}^{N_m}$
\begin{algorithmic}[1]
  \State Reconstruct $\omega$ from ~\eqref{EQ:omega Rec}
  \State Solve~\eqref{EQ:B System} for $\{F_1, \cdots, F_5\}$
  \State Reconstruct $\nu$ following~\eqref{EQ:nu Rec}
  \State Evaluate $\varphi:=\arctan\frac{\nu}{\omega}$
\end{algorithmic}
\end{algorithm}

\RED{In the uniqueness theory, we need intensity data from at least $N_m\ge d+2$ well-selected (to make the $\cK$ matrix invertible) $\kappa$ media to reconstruct the phase gradient $\nabla \varphi$ and data from at least $N_m\ge d+4$ well-selected (to ensure the invertibility of the matrix $\cB$) $\beta$ media to uniquely reconstruct the phase $\varphi$. Besides uniqueness consideration, one also needs to select media combinations such that the matrices $\cK$ and $\cB$ have good condition numbers. Therefore, it is better to obtain data from more media, say $N_m \gg d+4$, for better stability. In such cases, the linear systems~\eqref{EQ:K System} and~\eqref{EQ:B System} can be solved in the least-squares manner.}

Let us finish this section with the following important observations.
\begin{remark}
    The uniqueness proofs in~\Cref{SEC:LinearMedia} and~\Cref{SEC:NonlinearMedia}, as well as the reconstruction method derived from them, are essentially based on the so-called transport of intensity equations~\cite{ZuLiSuFaZhLuZhWaHuCh-OLE20}; see~\eqref{EQ:Ampl-Real-Imag} and~\eqref{EQ:Ampl-Real-Imag QNLS}. Standard ways of using this technique require boundary conditions for the phase, which is often impractical. Media diversity allows us to avoid such a requirement in our method. This is an essential difference between what we have here and what is in the literature.
\end{remark}

\section{Numerical experiments}
\label{SEC:Numer}

We now show some numerical experiments to demonstrate the performance of the analytic reconstruction procedure we outlined in the previous section.
\begin{figure}[!htb]
\centering
\begin{subfigure}{\textwidth}
\centering
\begin{subfigure}{0.33\textwidth}
    \includegraphics[width=\textwidth]{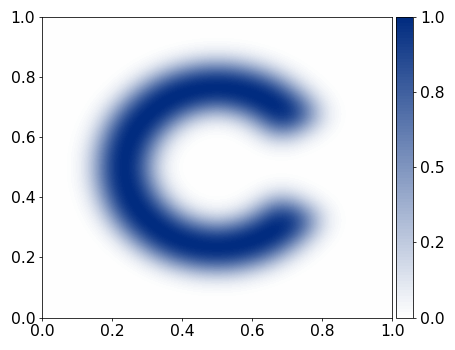}
    \caption{True phase}
\end{subfigure}
\hfill
\begin{subfigure}{0.66\textwidth}
    \includegraphics[width=0.49\textwidth]{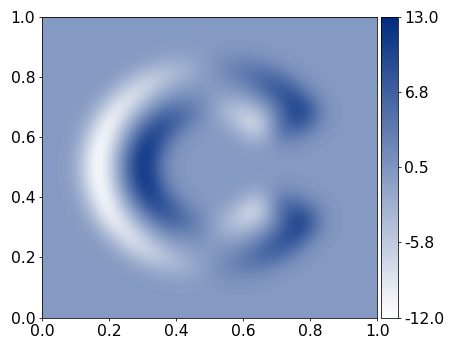}
    \includegraphics[width=0.49\textwidth]{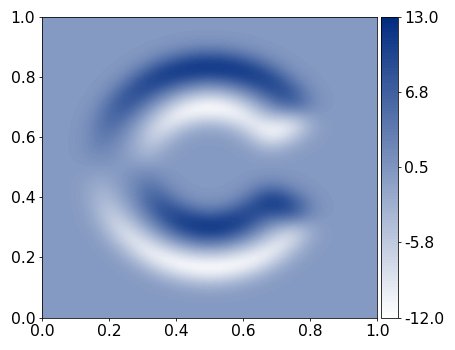}
    \caption{True $\omega \nabla \nu - \nu \nabla \omega$ ($x$ and $y$ components)}
\end{subfigure}
\end{subfigure}
\begin{subfigure}{\textwidth}
\begin{subfigure}{0.33\textwidth}
\hfill
\end{subfigure}
\begin{subfigure}{0.66\textwidth}
    \includegraphics[width=0.49\textwidth]{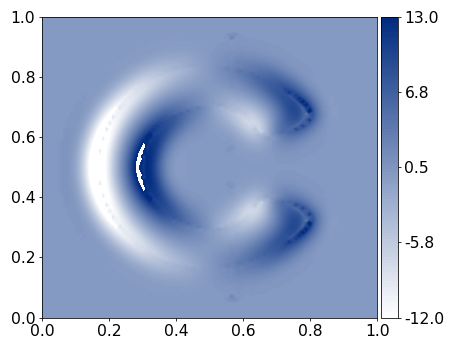}            \includegraphics[width=0.49\textwidth]{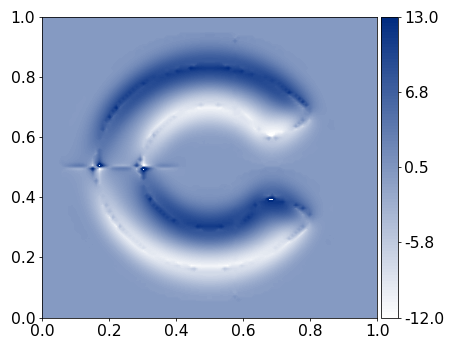}
    \caption{Reconstructed $\omega \nabla \nu - \nu \nabla \omega$ ($x$ and $y$ components) with $\kappa_1$ to $\kappa_4$}
\end{subfigure}
\end{subfigure}
\begin{subfigure}{0.33\textwidth}
\hfill
\end{subfigure}
\begin{subfigure}{0.66\textwidth}
    \includegraphics[width=0.49\textwidth]{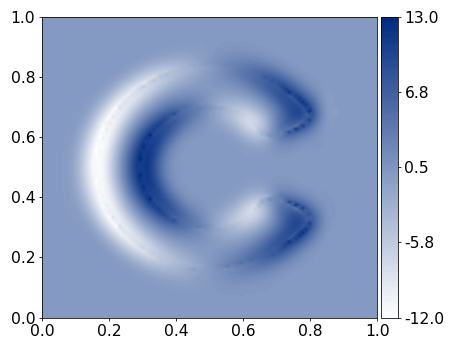}            \includegraphics[width=0.49\textwidth]{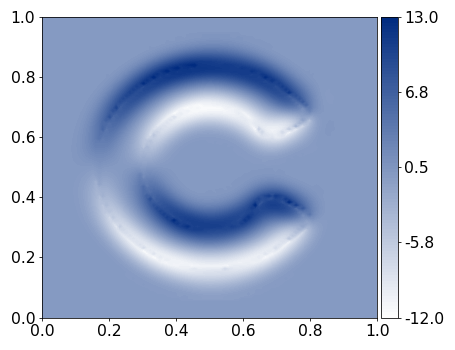}
    \caption{Reconstructed $\omega \nabla \nu - \nu \nabla \omega$ ($x$ and $y$ components) with $\kappa_1$ to $\kappa_5$}
\end{subfigure}
\caption{Phase gradient reconstruction in Experiment I by~\Cref{ALGO:Phase Grad Rec}.}
\label{FIG:Exp-1}
\end{figure}

\paragraph{Experiment I.} In the first numerical experiment, we consider phase gradient reconstructions using data from media with different $\kappa$ profiles. These profiles are defined, with $\bx=(x, y)$ as:
\[
\begin{array}{lll}
    \kappa_1(\bx)=1 + \sin(2\pi x)/2, &\kappa_2(\bx)=1 + \sin(2\pi y)/2,
    &\kappa_3(\bx)=1 + \cos(2\pi x)/2,\\[1ex]
    \kappa_4(\bx)=1 + \cos(2\pi x^2)/2,
    &\kappa_5(\bx)=1 + \cos(2\pi y)/2\,.
    & 
\end{array}
\]
These $\kappa$ profiles are selected to ensure the matrix $\cK$ in~\eqref{EQ:K} is full-rank when a sufficiently large number of $\kappa$ is used, as $\cK$ being invertible is a necessary condition for the reconstruction process. Two typical reconstructions are provided in~\Cref{FIG:Exp-1}. Shown are the true phase gradient (top row), the reconstruction with data from the first four media $\kappa_1\sim \kappa_4$ (middle row), and the reconstruction with data from all five media $\kappa_1\sim \kappa_5$.

The mild artifacts observed in~\Cref{FIG:Exp-1} (b) are due to the instability in solving the linear system \eqref{EQ:K System}. We can mitigate this instability by increasing the number of media, which makes the linear system over-determined. As demonstrated in~\Cref{FIG:Exp-1} (c), incorporating $\kappa_5$ into the media improves the quality of the reconstructions. We conducted several additional simulations. The quality of the reconstructions shown in~\Cref{FIG:Exp-1} is representative of our results.

In our example and in practice, it is inevitable that $\partial |u|/\partial z |_{z=0}$ is zero at some $\bx \in \Omega$, which creates a problem in the last step of \Cref{ALGO:Phase Grad Rec}. There are several workarounds: we can add a small $\epsilon$ to the denominator, or use a threshold approach where any denominator value below a certain small threshold is replaced by that threshold value. The approach we used in our reconstruction in \Cref{FIG:Exp-1} involves averaging the division in the local neighborhood. Specifically, when the denominator is zero, we compute the quotient by averaging the values of the division using the neighboring non-zero denominators. This method is reasonable when $\partial |u|/\partial z|_{z=0}$ is zero only in a negligibly small part of the domain.

\paragraph{Experiment II.} In the second numerical experiment, we show the reconstructions of the phase with data collected from media with nonlinearity controlled by $\beta$, following the study of~\Cref{SEC:NonlinearMedia}. The reconstructions are performed with~\Cref{ALGO:Phase Rec}. We considered data collected from seven different media whose $\beta$ profiles are given by
\begin{align*}
\beta_1(\bx) &= 1,              & \beta_2(\bx) &= 1 + \sin(2\pi x)/2,  & \beta_3(\bx) &= 1 + \cos(2\pi x)/2,   \\
\beta_4(\bx) &= 1 + \cos(2\pi y)/2, &\beta_5(\bx) &= 1 + \sin(2\pi x^2)/2, & \beta_6(\bx) &= 1 + \cos(2\pi y^2)/2, \\
 \beta_7(\bx) &= 1 + \sin(2\pi y)/2\,.   &              &
\end{align*}

Typical reconstruction results are presented in~\Cref{FIG:Exp6}. Shown from left to right are respectively the true phase profile, the phase profile reconstructed using data collected from media $\beta_1\sim \beta_6$, and the phase profile reconstructed using data collected from media $\beta_1\sim \beta_7$.~\Cref{THM:Unique} requires data from at least $6$ linearly independent $\beta$ profiles for the uniqueness of the reconstruction. This is the case of the first reconstruction shown in~\Cref{FIG:Exp6} (b). When we add data from an additional medium $\beta_7$, we only see a slight improvement in the reconstruction (see~\Cref{FIG:Exp6} (c)).

We have performed similar simulations on different phase profiles. The quality of the reconstructions is at the same level as shown in~\Cref{FIG:Exp6}. Adding data from more media does not seem to improve the reconstruction.
\begin{figure}[H]
    \centering
    \begin{subfigure}{0.30\textwidth}
        \includegraphics[width=\textwidth]{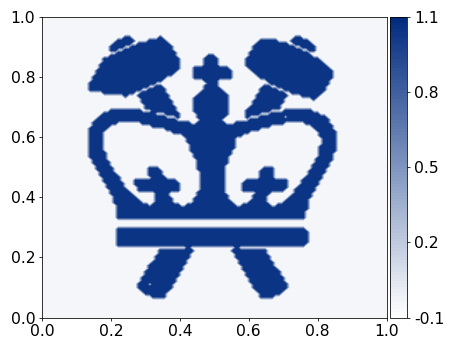}
        \caption{}
    \end{subfigure}
    \hfill
     \begin{subfigure}{0.30\textwidth}
        \includegraphics[width=\textwidth]{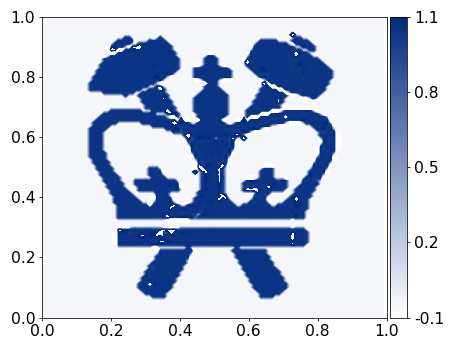}
        \caption{}
    \end{subfigure}
    \hfill
     \begin{subfigure}{0.30\textwidth}
        \includegraphics[width=\textwidth]{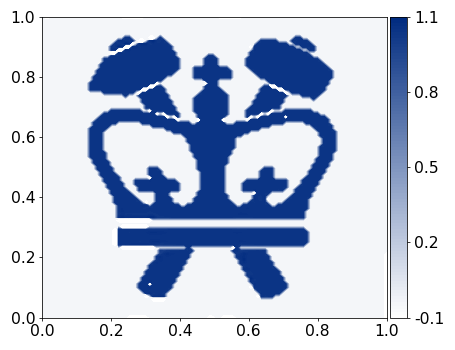}
        \caption{}
    \end{subfigure}    
    \caption{Phase reconstruction in Experiment II by \Cref{ALGO:Phase Rec}. Shown from left to right are (a) true phase, (b) phase reconstructed with data from $\beta_1\sim \beta_6$, and (c) phase reconstructed with data from $\beta_1\sim \beta_7$.}
    \label{FIG:Exp6}
\end{figure}

\RED{
There are several important aspects of the numerical implementation of the phase reconstruction algorithm.
}

\RED{First, the quality of the numerical reconstruction of the phase of the incoming wave remains on the same level for different intensity profiles, as long as $|f|\ge \eps>0$ for some $\eps$ much larger than the machine accuracy. For instance, we can have almost identical phase reconstructions as the ones 
 in~\Cref{FIG:Exp6} using either $|f(\bx)|=1$ or $|f(\bx)|=\exp(-\frac{(x-0.5)^2+(y-0.5)^2}{0.5})$.}

\RED{Second, instead of directly reconstructing the phase, the method we have reconstructs the (rescaled) real and imaginary parts of $f$, that is, $\nu$ and $\omega$ in the representation $f=a (\nu+i\omega)$ (where $a=|f|$). Our method does not explicitly enforce the constraint $\nu^2+\omega^2=1$ in the process. However, we observe that this constraint is satisfied very well for the successful reconstructions we performed. The reconstructions usually give $\nu^2+\omega^2=1$ up to a few digits of accuracy.
}

\RED{Third, the reconstruction methods we proposed are local, meaning that we need to invert matrices $\cK$ and $\cB$ at each point of the computational domain. These are small matrices whose inversion is computationally cheap to obtain. As we mentioned before, in our implementations, we take the simplest sets of media parameters $\kappa$ and $\beta$ to ensure that the matrices $\cK$ and $\cB$ are invertible at every point in the computational domain.} 

\section{Computational retrieval in general settings}
\label{SEC:Generalizations}

We now turn our attention to numerical studies, specifically focusing on the phase retrieval problem in contexts that extend beyond the discussions in the previous sections. Given measured the intensity data at a sequence of $N_p$ transverse planes located at $\{z_p\}_{p=1}^{N_p}$, for a group of $N_m$ different media $\{\kappa_m, \beta_m\}_{m=1}^{N_m}$, we search for the incident field $f(\bx)$ as the minimizer of the following least-squares data mismatch functional:
\begin{equation}\label{EQ:Min Func}
\Phi(f):= \frac{1}{2} \sum_{m=1}^{N_m} \sum_{p=1}^{N_p} \int_\Omega \big(|u_{m,p}|^2-d_{m,p}^2\big)^2 d\bx
\end{equation}
where $u_{m,p}(\bx)=u_m(z_p,\bx)$ is the wave field solution corresponding to medium $m$ measured at the transverse plane $z_p$, and $d_{m,p}$ is the corresponding measured intensity data. In the following experiments, $L_z=0.02$ and the multi-plane data are measured at $10$ transverse planes located at $z_p = 0.01 + 0.001p$, $p=0,\dots,9$.

The computations in this section are all performed in the spatial domain $\Omega=(0,L_x)\times (0,L_y)$. We impose periodic boundary conditions on the solutions to the Schr\"{o}dinger equations~\eqref{EQ:LS},~\eqref{EQ:QNLS}, and~\eqref{EQ:GP}. The optimization problem is solved iteratively through a quasi-Newton optimization method, where the gradient, or the Fr\'echet derivative, is computed via the adjoint state method. The relevant calculations and the implementation details are documented in~\Cref{SEC:Implementation}. 

\RED{Let us emphasize again that while our main interest is the phase of the incident wave, our reconstruction algorithm attempts to reconstruct the incident wave $f$ itself. In other words, if $f(\bx) = a(\bx)\exp(i\varphi(\bx))$, our algorithm reconstruct both the amplitude $|f|$ and the phase $\varphi$. If we also measure the amplitude of $f$, $a(\bx)$, we can add $|f|=a$ as a constraint in the optimization process. In such a case, it would be more convenient to view $\Phi$ as a functional $\Psi$ of $\varphi$, i.e., 
\begin{equation}\label{EQ:Psi}
    \Psi(\varphi) = \Phi(a\exp(i\varphi)).    
\end{equation}
The Fr\'echet derivative of $\Psi$ can be computed using the Fr\'echet derivative of $\Psi$ along with the chain rule; more details can be found in~\Cref{SEC:Implementation}.}

Since our numerical experiments are also performed using synthetic data, we always know the true phase $\varphi_{\rm true}$. To quantify the error in the reconstruction $\varphi_{\rm reconstructed}$, we define its mean shift from the truth by
\begin{equation}\label{EQ:shift}
    c^*=\argmin_{c \in \bbR}\|\varphi_{\text{reconstructed}}+c-\varphi_{\text{true}}\|_{L^2(\Omega)}^2\,.
\end{equation}
We then define the relative error of the reconstruction, after the correction of the phase shift, as
\begin{equation}\label{EQ:rel_err}
    \cE= \dfrac{\|\varphi_{\text{reconstructed}}-\varphi_{\text{true}} - c^*\|_{L^2(\Omega)}}{\|\varphi_{\text{true}}\|_{L^2(\Omega)}}\,.
\end{equation}

\subsection{Single medium data}

We start with numerical reconstructions using data from only a single medium. We consider both the case of single-plane data and the case of multi-plane data. In the multi-plane data case, we use data collected in planes located at $\{z_p\}_{p=1}^{N_p}$. In the single-plane data case, we use only data from the last plane located at $z_{N_p}$.
\begin{figure}[!htp]
    \centering
    \begin{subfigure}{0.24\textwidth}
        \includegraphics[width=\textwidth]{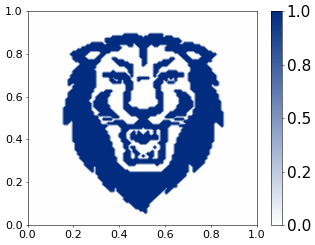}
        \captionsetup{font=footnotesize}
        \caption{}
    \end{subfigure}
     \begin{subfigure}{0.24\textwidth}
        \includegraphics[width=\textwidth]{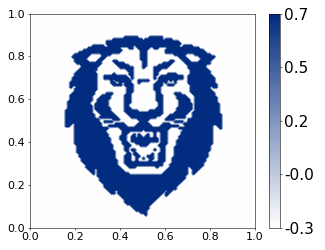}
        \captionsetup{font=footnotesize}
        \caption{}
    \end{subfigure}
     \begin{subfigure}{0.24\textwidth}
        \includegraphics[width=\textwidth]{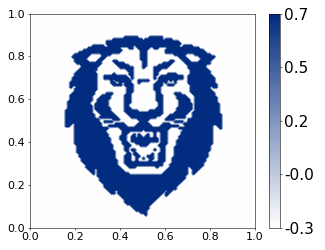}
        \captionsetup{font=footnotesize}
        \caption{}
    \end{subfigure}
     \begin{subfigure}{0.24\textwidth}
        \includegraphics[width=\textwidth]{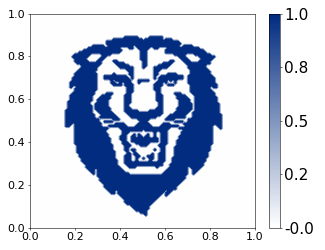}
        \captionsetup{font=footnotesize}
        \caption{}
    \end{subfigure}
    \begin{subfigure}{0.24\textwidth}
    \hfill
    \end{subfigure}
     \begin{subfigure}{0.24\textwidth}
        \includegraphics[width=\textwidth]{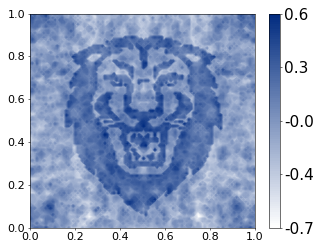}
        \captionsetup{font=footnotesize}
        \caption{}
    \end{subfigure}
     \begin{subfigure}{0.24\textwidth}
        \includegraphics[width=\textwidth]{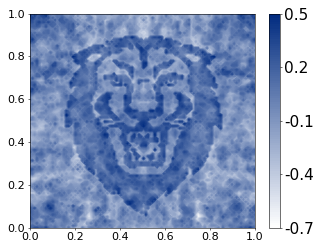}
        \captionsetup{font=footnotesize}
        \caption{}
    \end{subfigure}
     \begin{subfigure}{0.24\textwidth}
        \includegraphics[width=\textwidth]{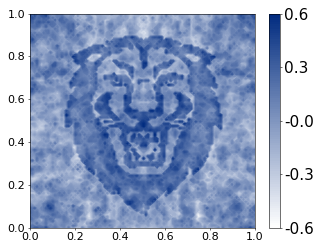}
        \captionsetup{font=footnotesize}
        \caption{}
    \end{subfigure} 
    \caption{Phase reconstructions with multi-plane (top row) and single-plane (bottom row) data from a single medium. Left to right: (a) true phase, (b, e) phase reconstructed with the LSE model~\eqref{EQ:LS}, (c, f) phase reconstructed with the GP model~\eqref{EQ:GP}, and (d, g) phase reconstructed with the QNLS model~\eqref{EQ:QNLS}. 
    }
    \label{FIG: Single Media}
\end{figure}

A typical set of reconstructions for the linear Schr\"{o}dinger model~\eqref{EQ:LS}, the quadratic nonlinear Schr\"{o}dinger model~\eqref{EQ:QNLS}, and the Gross-Pitaevskii model~\eqref{EQ:GP} is presented in~\Cref{FIG: Single Media}, and the quality of the reconstructions, measured in the mean phase shift $c^*$ and the relative error $\cE$ are summarized in~\Cref{TAB:Error Single-Medium}. 

There are two obvious features that we can observe from the reconstructions. The first is that reconstructions from multi-plane data are significantly more accurate than those from only single-plane data. This is true for all models. Relative errors are negligible for reconstructions with multi-plane data but on the order of $6\times 10^{-1}$ for single-plane data reconstructions.
\begin{table}[!htp]
    \centering
    \renewcommand{\arraystretch}{1.25} 
    \setlength\tabcolsep{7pt}       
    \begin{tabular}{|c|c|c|c|c|}
    \hline\hline
    {\bf Model} & {\bf Media} & {\bf Data} & {\bf Phase Shift $c^*$} & {\bf Relative Error $\cE$} \\
    \hline\hline
    \multirow{2}{*}{LS ~\eqref{EQ:LS}} & \multirow{2}{*}{$\kappa = 1$} & 
    multi-plane & -0.278 & 7.85e-09 \\
    \cline{3-5}
    & & single-plane & -0.283 & 6.13e-01 \\
    \hline\hline
    \multirow{2}{*}{GP~\eqref{EQ:GP}} & \multirow{2}{*}{$\begin{matrix}
        \kappa=1 \\
        \beta= 11 + 10\cos(10 \pi x)
    \end{matrix}$} & 
    multi-plane & -0.283 & 8.96e-09 \\
    \cline{3-5}
    & & single-plane & -0.305 & 6.91e-01 \\
    \hline\hline
    \multirow{2}{*}{QNLS~\eqref{EQ:QNLS}} & \multirow{2}{*}{$\begin{matrix}
        \kappa=1 \\
        \beta= 11 + 10\cos(10 \pi x)
    \end{matrix}$} & 
    multi-plane & -3.74e-10 & 2.27e-08 \\
    \cline{3-5}
    & & single-plane & -0.27 & 6.76e-01 \\
    \hline\hline
    \end{tabular}
    \caption{Phase shifts and relative errors for reconstructed phases from~\Cref{FIG: Single Media}.
    }
    \label{TAB:Error Single-Medium}
\end{table}

The second obvious observation is that a significant phase shift is observed in all reconstructions besides the multi-plane data reconstruction with the quadratic nonlinear Schr\"{o}dinger model~\eqref{EQ:QNLS}. While the observed constant shift is approximately $c^*=-0.28$ for the cases shown here, this specific value of $c^*$ holds no unique significance. Our extra experiments indicate that altering the initial guess of the optimization algorithm results in varied phase shifts in the reconstructions.

The fact that reconstructions with the quadratic nonlinear Schr\"{o}dinger model~\eqref{EQ:QNLS} using data from only a single medium has no phase ambiguity does not contradict the examples of nonuniqueness that we constructed in~\Cref{SEC:Nonuniqueness}, as those constructions only work for specific incident waves that make the solutions independent of $z$. The numerical reconstructions here (and the many numerical experiments we have done but did not present here) seem to suggest that it is generally possible to uniquely reconstruct the phase of the incident wave from multi-plane intensity measurement for the quadratic Schr\"{o}dinger model~\eqref{EQ:QNLS}.

\subsection{Multi-media data}

In the next group of numerical simulations, we study phase reconstructions using data from multiple media. We again consider both single-plane and multi-plane data.
\begin{figure}[!htb]
    \centering 
     \begin{subfigure}{0.30\textwidth}        
        \includegraphics[width=\textwidth]{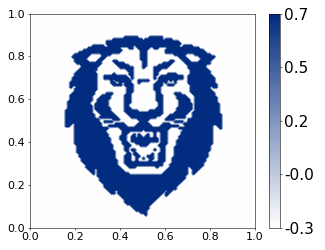}
        \captionsetup{font=footnotesize}
        \caption{}
    \end{subfigure}
     \begin{subfigure}{0.30\textwidth}
        \includegraphics[width=\textwidth]{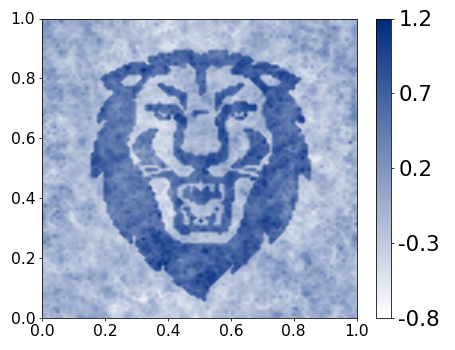}
        \captionsetup{font=footnotesize}
        \caption{}
    \end{subfigure}

     \begin{subfigure}{0.30\textwidth}
        \includegraphics[width=\textwidth]{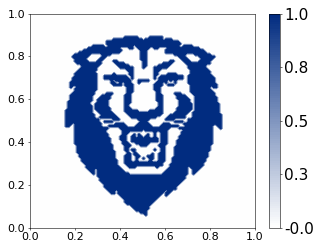}
        \captionsetup{font=footnotesize}
        \caption{}
    \end{subfigure}
     \begin{subfigure}{0.30\textwidth}
        \includegraphics[width=\textwidth]{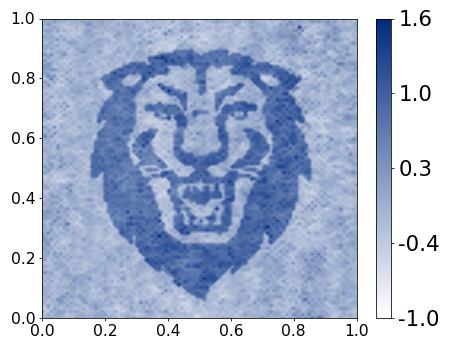}
        \captionsetup{font=footnotesize}
        \caption{}
    \end{subfigure}   
    \caption{Phase reconstructions for the GP model~\eqref{EQ:GP} (top) and QNLS model~\eqref{EQ:QNLS} (bottom) using data from two different media $(\kappa_1,\beta_1)$ and $(\kappa_2,\beta_2)$ in~\eqref{EQ:Kappa-Beta}. Shown are reconstructions with multi-plane (left) and single-plane (right) data.}
    \label{FIG: 2 beta kappa}
\end{figure}

Typical reconstructions are presented in~\Cref{FIG: 2 beta kappa} and~\Cref{FIG: 3 beta kappa}. The quality of the reconstruction is summarized in~\Cref{TAB:Error-Multi-Media}. The media profiles we used are
\begin{equation}\label{EQ:Kappa-Beta}
\begin{array}{rcl}
(\kappa_1,\beta_1) &=& (1, 11 + 10 \cos(10 \pi x)),\\
(\kappa_2,\beta_2) &=& (1 + 0.5 \cos(5 \pi x), 11 + 10 \cos (10 \pi y)),\\
(\kappa_3,\beta_3) &=& (1 + 0.5 \cos(5 \pi y), 11 + 10 \sin(10 \pi x))\,.
\end{array}
\end{equation}
In the reconstructions of~\Cref{FIG: 2 beta kappa}, we used the two media $(\kappa_1, \beta_1)$ and $(\kappa_2, \beta_2)$ while in the results of~\Cref{FIG: 3 beta kappa}, we used all three media of~\eqref{EQ:Kappa-Beta}. The initial guess for the phase for the optimization algorithm is $\phi_0=0$. Neither the particular media profile nor the initial guess makes a substantial difference in the reconstruction results.
\begin{figure}[!htb]
    \centering 
     \begin{subfigure}{0.30\textwidth}
        \includegraphics[width=\textwidth]{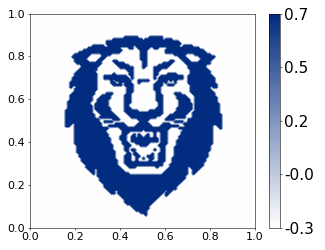}
        \captionsetup{font=footnotesize}
        \caption{}
    \end{subfigure}
     \begin{subfigure}{0.30\textwidth}
        \includegraphics[width=\textwidth]{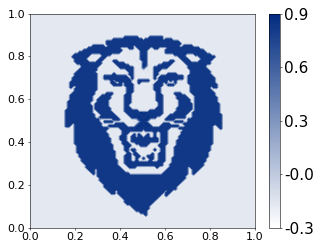}
        \captionsetup{font=footnotesize}
        \caption{}
    \end{subfigure}

     \begin{subfigure}{0.30\textwidth}
        \includegraphics[width=\textwidth]{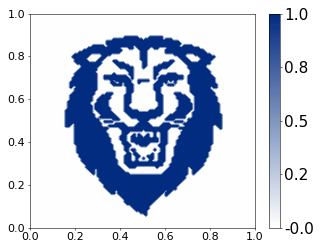}
        \captionsetup{font=footnotesize}
        \caption{}
    \end{subfigure}
     \begin{subfigure}{0.30\textwidth}
        \includegraphics[width=\textwidth]{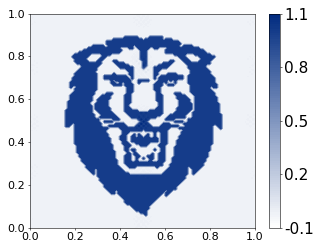}
        \captionsetup{font=footnotesize}
        \caption{}
    \end{subfigure}
    
    \caption{Same as~\Cref{FIG: 2 beta kappa} but with data from three media $(\kappa_1,\beta_1)$, $(\kappa_2,\beta_2)$, and $(\kappa_3,\beta_3)$ in~\eqref{EQ:Kappa-Beta}.}
    \label{FIG: 3 beta kappa}
\end{figure}

The key observations from this group of numerical simulations are as follows. First, using multiple media data does not help us eliminate the constant phase shift in the reconstructions with the Gross-Pitaevskii model~\eqref{EQ:GP}, which is expected from the $u\mapsto ue^{i\varphi}$ invariance of the solution to the equation. Second, using data from multiple media can significantly improve the phase reconstructions (which is very clear by a visual comparison between~\Cref{FIG: 2 beta kappa} and~\Cref {FIG: 3 beta kappa}). Third, when more and more media are used, even data at a single transverse plane seems enough to remove phase ambiguity for the quadratic nonlinear Schr\"{o}dinger model.
\begin{table}[!htb]
    \centering
    \renewcommand{\arraystretch}{1.25} 
    \setlength\tabcolsep{7pt}       
    \begin{tabular}{|c|c|c|c|c|}
    \hline\hline
    {\bf Model} & {\bf Media} & {\bf Data} & {\bf Phase Shift $c^*$} & {\bf Relative Error $\cE$} \\
    \hline\hline
    \multirow{2}{*}{GP~\eqref{EQ:GP}} & \multirow{4}{*}{
    $\{\kappa_m, \beta_m)\}_{m=1}^2$ in~\eqref{EQ:Kappa-Beta}
    } &  
    multi-plane & -0.267 & 4.27e-10 \\
    \cline{3-5}
    & & single-plane & -0.229 & 3.95e-01 \\
    \cline{1-1}\cline{3-5}
    \multirow{2}{*}{QNLS~\eqref{EQ:QNLS}} & & 
    multi-plane & 3.26e-10 & 4.04e-10 \\
    \cline{3-5}
    & & single-plane & -0.147 & 4.05e-01 \\
    \cline{1-5}
    \multirow{2}{*}{GP~\eqref{EQ:GP}} & \multirow{4}{*}{
    $\{\kappa_m, \beta_m)\}_{m=1}^3$ in~\eqref{EQ:Kappa-Beta}
    } & 
    multi-plane & -0.259 & 6.17e-11 \\
    \cline{3-5}
    & & single-plane & -0.246 & 4.14e-04 \\
    \cline{1-1}\cline{3-5}
    \multirow{2}{*}{QNLS~\eqref{EQ:QNLS}} & & 
    multi-plane & -9.77e-11 & 2.76e-10 \\
    \cline{3-5}
    & & single-plane & 6.27e-05 & 7.08e-04 \\
    \cline{1-5}
    \hline\hline
    \end{tabular}
    \caption{Phase shifts and relative errors for reconstructed phases in~\Cref{FIG: 2 beta kappa} and~\Cref{FIG: 3 beta kappa}. 
    }
    \label{TAB:Error-Multi-Media}
\end{table}

To further investigate whether or not multiple media data at a single transverse plane is sufficient to uniquely reconstruct the phase, we perform another set of numerical experiments where we use data from seven sets of media given by
\[
(\kappa_m, \beta_m) = (1 + 0.5 \sin(4 \pi m y),  11+10 \cos (6\pi (1+m) x)), \quad m=0,1,\cdots, 6\,.
\]
We started the optimization algorithm from different initial guesses of the phase. The algorithm always converges to the same solution with no phase shift (up to the numerical error). Three of the reconstructions are shown in~\Cref{FIG: terminal}. This clearly demonstrates that a single plane measurement from a sufficiently large number of media is enough to uniquely determine the phase of the incident wave. A mathematical theory for this observation is currently missing.
\begin{figure}[!htb]
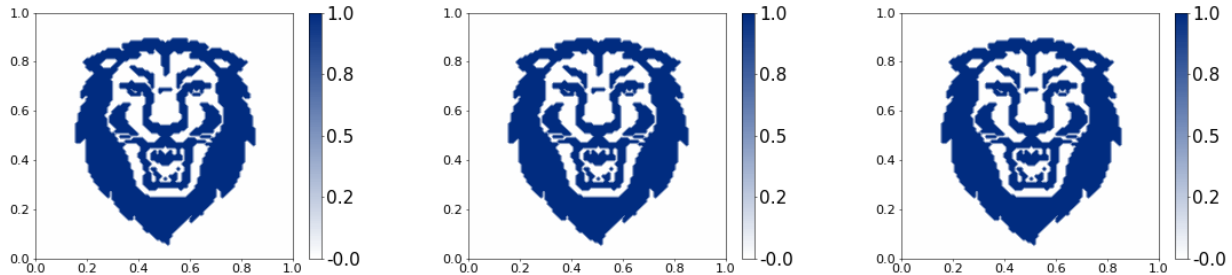

    \centering
    \begin{subfigure}{0.30\textwidth}
        \includegraphics[width=\textwidth]{Figures/PR_3/phase_0.png}
        \captionsetup{font=footnotesize}
    \end{subfigure}
    \hfill
     \begin{subfigure}{0.30\textwidth}
        \includegraphics[width=\textwidth]{Figures/PR_3/phase_1.png}
        \captionsetup{font=footnotesize}
    \end{subfigure}
    \hfill
     \begin{subfigure}{0.30\textwidth}
        \includegraphics[width=\textwidth]{Figures/PR_3/phase_2.png}
        \captionsetup{font=footnotesize}
    \end{subfigure}

    \caption{Phase reconstructions for the quadratic nonlinear Schr\"{o}dinger model~\eqref{EQ:QNLS} with single-plane, multiple media data. Shown are reconstructions by an optimization algorithm with three different randomly selected initial guesses.}
    \label{FIG: terminal}
\end{figure}

\subsection{Impact of measurement noise}

In the last set of numerical experiments, we examine the stability of the phase retrieval problem with respect to random noise in the intensity data. We generated noisy data by polluting the true data with multiplicative Gaussian noise whose strength is controlled by a parameter $\alpha$. More precisely, noisy data $\wt d$ are generated from $d$ using the relation
\[
\wt d(\bx) = (1 + \alpha \mathcal{N}(0,1))\, d(\bx)
\]
$\cN(0, 1)$ denotes a Gaussian random field with mean $0$ and variance $1$.

Phase reconstructions from typical realizations of the noisy data are shown in~\Cref{FIG: noise T1T2} for the quadratic nonlinear Schr\"{o}dinger model with data from two different media. While in general we do observe that the reconstructions are significantly affected by the noise in the data, it is also clear that the reconstructions are fairly stable. Even in the absence of regularization, the reconstruction algorithm converges to decent results. Extensive numerical simulations with different media profiles show the same characteristics observed in~\Cref{FIG: noise T1T2}.
\begin{figure}[!htb]
    \centering
    \begin{subfigure}{0.30\textwidth}
        \includegraphics[width=\textwidth]{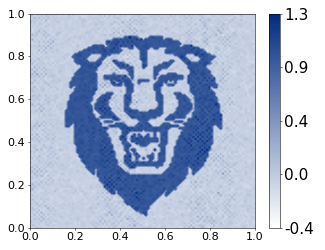}
        \caption{$\alpha = 0.1$}
    \end{subfigure}
    \hfill
    \begin{subfigure}{0.30\textwidth}
        \includegraphics[width=\textwidth]{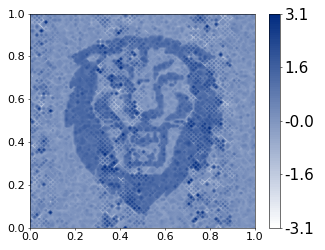}
        \caption{$\alpha = 0.5$}
    \end{subfigure}
    \hfill
    \begin{subfigure}{0.30\textwidth}
        \includegraphics[width=\textwidth]{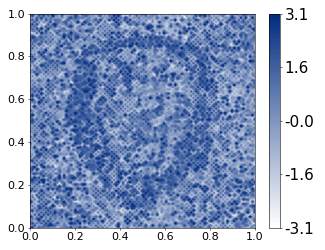}
        \caption{$\alpha = 1$}
    \end{subfigure}
    \begin{subfigure}{0.30\textwidth}
        \includegraphics[width=\textwidth]{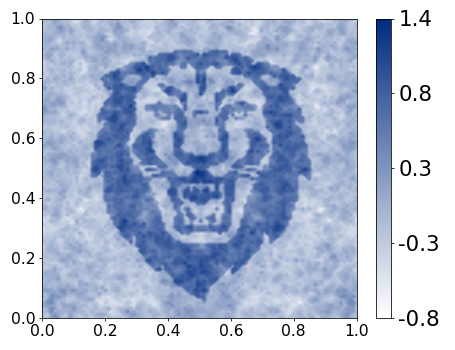}
        \caption{$\alpha = 0.001$}
    \end{subfigure}
    \hfill
    \begin{subfigure}{0.30\textwidth}
        \includegraphics[width=\textwidth]{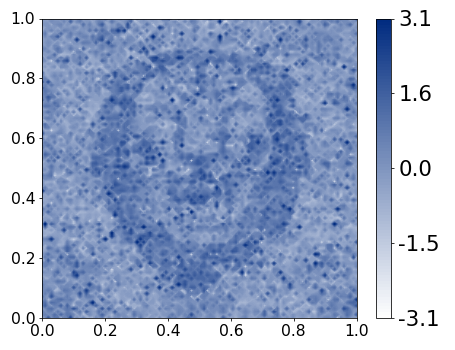}
        \caption{$\alpha = 0.01$}
    \end{subfigure}
    \hfill
    \begin{subfigure}{0.30\textwidth}
        \includegraphics[width=\textwidth]{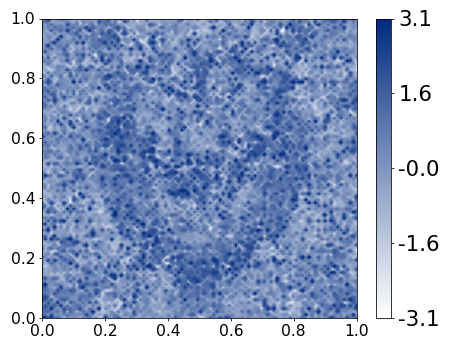}
        \caption{$\alpha = 0.02$}
    \end{subfigure}

    \caption{Typical phase reconstructions with the QNLS model~\eqref{EQ:QNLS} using noisy multi-plane (top row) and single-plane (bottom row) data. 
    }
\label{FIG: noise T1T2}
\end{figure}

In general, having multi-plane data always helps average out noise in the data so that the effective noise strength is lower than that presented by the same $\alpha$ in the single-plane case. This is why the multi-plane data reconstructions are much better than their single-plane data counterparts, as shown in~\Cref{FIG: noise T1T2}.

\section{Concluding remarks}
\label{SEC:Concl}

This work studies phase retrieval for wave fields using linear and nonlinear media diversity, aiming to recover the phase of an incoming wave from intensity measurements behind multiple media. While such a phase reconstruction is generally impossible, we show that it can be uniquely achieved when we have a well-selected group of media. Moreover, when this is achievable, direct (i.e. non-iterative) methods can be developed for the reconstruction of the phase from measured intensity data. While our derivation of the phase reconstruction method is based on the transport of intensity equation, our method does NOT require boundary conditions for the phase, making it dramatically different from similar methods in the literature that are based on the transport of intensity equations.

We also demonstrate numerically that phase retrieval is feasible from multi-plane and multi-media intensity data outside of the regime where our analytic derivations are possible.

There are many theoretical and computational issues to be solved in the future. One particularly interesting question is whether multi-plane data is enough for the unique recovery of the phase with a certain equivalence class. As we have shown with examples, for a given medium, one can always find an incident wave whose phase is not uniquely recoverable. However, our numerical simulations show good phase reconstruction results for other incident waves. A better mathematical understanding of this issue is needed. Another interesting question to ask is whether or not a single medium governed by the quadratic nonlinear Schr\"{o}dinger model~\eqref{EQ:QNLS} is enough to achieve uniqueness (modulo the equivalence class of incident waves that we constructed as counterexamples in~\Cref{SEC:Nonuniqueness}). From the practical point of view, there is still a lot to learn on how to maximize phase information in intensity data by tuning simultaneously $\kappa$ and $\beta$ of the underlying media.

\section*{Acknowledgment}

This work is partially supported by the National Science Foundation through grants DMS-1937254 and DMS-2309802. The computational experiments conducted in this work are made possible with the support of the Research Computing Services (RCS) at Columbia University.

\appendix 

\section{Examples of non-uniqueness without media diversity}
\label{SEC:Nonuniqueness}

Here we construct some examples where one cannot uniquely reconstruct the phase of the incident wave from data measured in one nonlinear medium only, even if the measurements are taken at multiple transverse planes. While most of these examples are for special situations, they show that we need to be careful about what we need to measure to ensure the success of phase retrieval.

\subsection{Non-uniqueness for the quadratic NLS}

We start with a one-dimensional example of a particular $\beta$ for which non-uniqueness occurs for the quadratic nonlinear Schr\"odinger model~\eqref{EQ:QNLS}.

\begin{proposition}\label{Prop:Nonunique 1D QNLS}
Let $\Omega=\bbR$, $\kappa(x)=1$, and $\beta(x) = 4\sech(2x) = \frac{8}{e^{2x} + e^{-2x}}$. Then there exist two distinct incident waves $f_1(x)$ and $f_2(x)$ such that the corresponding solutions $u_1(z,x)$ and $u_2(z,x)$ of~\eqref{EQ:QNLS} satisfy $|u_1(z,x)| = |u_2(z,x)|$ for all $(z,x)\in (0, \infty)\times \Omega$.
\end{proposition}
\begin{proof} Define the function $\psi(x) := 2\arctan(\tanh(x))$. Direct calculation shows that 
\begin{align*}
    \beta(x) = \frac{|\psi'(x)|^2}{\cos \psi(x)} = -\frac{\psi''(x)}{\sin \psi(x)}.
\end{align*}
More calculation shows that 
\begin{align*}
    u_+(z,x) := e^{i\psi(x)},\qquad  u_-(z,x) := e^{-i\psi(x)}
\end{align*}
are solutions to~\eqref{EQ:QNLS}. These solutions have the property that 
\begin{align*}
    |u_+(z,x)| = |u_-(z,x)| = 1,\qquad \forall\,(z,x)\in (0, \infty)\times \Omega,
\end{align*}
yet their initial conditions 
\begin{equation*}
    f_+(x) := u_+(0,x) = e^{i\psi(x)},\qquad f_-(x) := u_-(0,x) = e^{-i\psi(x)}
\end{equation*}
are different.
\end{proof}

\begin{proposition}
Let $\psi \in \cC^2(\overline\Omega; \bbR)$ be a function satisfying the partial differential equation
$$\Delta \psi = -(\tan \psi)|\nabla \psi|^2\,.$$
We define $\beta(z, \bx): (0, \infty)\times \Omega \to \mathbb{R}$ and $\kappa(z, \bx): (0, \infty)\times \Omega \to \mathbb{R}$ such that
$$
\frac{\beta(z,\bx)}{\kappa(z,\bx)}:= \frac{|\nabla\psi|^2}{\cos \psi} = -\frac{\Delta\psi}{\sin \psi}\,.
$$
Then, the functions $u_{\pm}: (0,\infty) \times \Omega \to \mathbb{C}$ defined by:
$$u_{\pm}(z,\bx) := e^{\pm i\psi(\mathbf{x})}$$
are solutions to~\eqref{EQ:QNLS}. Moreover, $u_+$ and $u_-$ have the same amplitude but different initial conditions.
\end{proposition}
\begin{proof}
Let $u_{\pm}(z, \bx) := e^{\pm i\psi(x)}$, then
\begin{align*}
    \Delta u_+ & = i \Delta \psi e^{i\psi} - |\nabla \psi|^2 e^{i\psi} = -(1+i\tan \psi)|\nabla \psi|^2 e^{i\psi} 
\end{align*}
\begin{align*}
    \Delta u_- & = -i \Delta \psi e^{-i\psi} - |\nabla \psi|^2 e^{-i\psi} = (i\tan \psi - 1)|\nabla \psi|^2 e^{-i\psi} 
\end{align*}
where we used $\Delta \psi = -(\tan \psi)|\nabla \psi|^2$, and
\begin{align*}
    \frac{\beta}{\kappa} u_+^2 &= \frac{|\nabla\psi|^2}{\cos \psi} e^{i 2\psi} = \frac{e^{i \psi} }{\cos \psi} |\nabla \psi|^2 e^{i\psi} =  \frac{\cos \psi + i\sin \psi}{\cos \psi} |\nabla \psi|^2 e^{i\psi} = (1+i\tan \psi)|\nabla \psi|^2 e^{i\psi} 
\end{align*}
and
\begin{align*}
    \frac{\beta}{\kappa} u_-^2 &= \frac{|\nabla\psi|^2}{\cos \psi} e^{-i 2\psi} = \frac{e^{-i \psi} }{\cos \psi} |\nabla \psi|^2 e^{-i\psi} =  \frac{\cos \psi - i\sin \psi}{\cos \psi} |\nabla \psi|^2 e^{-i\psi} = (1-i\tan \psi)|\nabla \psi|^2 e^{-i\psi} .
\end{align*}
This shows $i\partial u_\pm /\partial z +\kappa \Delta u_\pm +\beta u_\pm^2 = 0$. Indeed, $|u_+| = |u_-| = |e^{\pm i\psi}| = 1$ so they have the same magnitudes, yet $u_+(0,\bx) = e^{i\psi(\bx)}$ and $u_-(0,\bx) = e^{-i\psi(\bx)}$ so their initial conditions are distinct in general. \qedhere	
\end{proof}

The following general construction is straightforward to verify.
\begin{proposition}\label{PROP:General Nonunique}
For any $\kappa(z,\bx)=\kappa(\bx)$ and $\beta(z, \bx)=\beta(\bx)$ such that $\dfrac{\beta}{\kappa}(\bx)\neq 0$, let $\varphi(\bx)\in \cC^2(\Omega; \bbC)$ be a complex-valued function that satisfies
\begin{equation}\label{EQ:Stationary Sol}
    \Delta \varphi+\frac{\beta}{\kappa}(\bx) \varphi^2=0,\qquad \mbox{in}\ \ \Omega\,.
\end{equation}
Then 
\[
    u_+(z, \bx)=\varphi(\bx) \quad \mbox{and}\quad u_-(z, \bx)=\overline{\varphi(\bx)} 
\]
are solutions to~\eqref{EQ:QNLS} with incident waves $u_+(0, \bx):=\varphi$ and $u_-(0, \bx):=\overline{\varphi}$ respectively, and $|u_+|=|u_-|$.
\end{proposition}
The existence and nonuniqueness of solutions to~\eqref{EQ:Stationary Sol} are classical subjects of studies; see, for instance, ~\cite{AmPr-AMPA72, BaRe-CM11} and references therein.

\subsection{Non-uniqueness for the GP equation}

Similar nonuniqueness examples can be constructed for the GP model~\eqref{EQ:GP}. Recall that we can only hope to reconstruct the initial condition $f$ up to a constant phase shift. Therefore, we henceforth identify initial conditions that differ by a constant phase and consider the problem of reconstructing the equivalence class of $f$ under this equivalence relation.

\begin{proposition} \label{Prop:GP nonunique}
Let $\kappa(z, \bx)=1$ and $\beta(z, \bx) = \beta(z)$ (that is, $\beta$ is independent of the space variable). Then there exist two initial conditions $f_1$ and $f_2$ belonging to different equivalence classes such that the corresponding solutions $u_1$ and $u_2$ of the Gross-Pitaevskii equation \eqref{EQ:GP} satisfy $|u_1(z, \bx)| = |u_2(z, \bx)|$ for all $(z, \bx)\in (0, \infty)\times \Omega$.
\end{proposition}

\begin{proof} Let $\bk_1 = ((\bk_1)_x, (\bk_1)_y)\in \bbR^2$ and $\bk_2 = ((\bk_2)_x, (\bk_2)_y)\in \bbR^2$ be two distinct vectors such that $(\bk_i)_x \in \frac{2\pi}{L_x}\bbZ$ and $(\bk_i)_y \in \frac{2\pi}{L_y}\bbZ$ for $i=1,2$. Then 
\begin{align*}
    u_1(z, \bx) &:= \exp\left(i\left[\bk_1\cdot \bx - |\bk_1|^2 z + \int_0^z \beta(s)\,ds \right]\right),\\
    u_2(z, \bx) &:= \exp\left(i\left[\bk_2\cdot \bx - |\bk_2|^2 z + \int_0^z \beta(s)\,ds \right]\right)
\end{align*}
are solutions to the GP equation \eqref{EQ:GP} with periodic boundary conditions, as is readily verified. These solutions have the property that 
\begin{align*}
    |u_1(z, \bx)| = |u_2(z, \bx)| = 1,\qquad \forall\,(z, \bx)\in (0,\infty)\times \Omega,
\end{align*}
yet their initial conditions 
\begin{equation*}
    f_1(\bx) := u_1(0,\bx) = e^{i\bk_1\cdot\bx},\qquad f_2(\bx) := u_2(0,\bx) = e^{i\bk_2\cdot\bx}
\end{equation*}
belong to different equivalence classes.
\end{proof}
 
To focus on the main idea, we set $\Omega=\bbR^2$ in order not to have to deal with boundary conditions.  Let us define a subset $\cB' \subset \cB$ as follows. We say that $\beta=\beta(\bx)\in \cB$ is in $\cB'$ if: there exists a constant $c$ and a non-constant harmonic function $\psi\in \cC^2(\overline{\Omega}; \bbR)$ such that 
 \begin{align}\label{EQ:B'}
     \beta(\bx) = c+|\nabla \psi(\bx)|^2,\qquad \forall\,\bx\in\Omega.
 \end{align}
 For instance, all constant functions belong to the class $\cB'$, as can be seen by choosing $\psi$ of the form $\psi(\bx) = \bk\cdot \bx$ in the definition. Another example of a function belonging to $\cB'$ is the function $\beta(x,y) = x^2+y^2$, as can be seen by choosing $c=0$ and the harmonic function $\psi(x,y) = (x^2-y^2)/2$. 

\begin{proposition} Suppose $\beta=\beta(\bx)$ lies in $\cB'$. Then there exist two initial conditions $f_1$ and $f_2$ belonging to different equivalence classes such that the corresponding solutions $u_1$ and $u_2$ of the Gross-Pitaevskii equation \eqref{EQ:GP} satisfy $|u_1(z, \bx)| = |u_2(z, \bx)|$ for all $(z, \bx)\in (0, \infty) \times \Omega$.
\end{proposition}

\begin{proof} Let $c$ be $\psi$ be as in \eqref{EQ:B'}. One readily checks through a straightforward calculation that 
\begin{align*}
    u_+(z, \bx) &:= e^{i(\psi(\bx) + cz)},\\
    u_-(z, \bx) &:= e^{i(-\psi(\bx) + cz)}
\end{align*}
are solutions to~\eqref{EQ:GP}. These solutions have the property that 
\begin{align*}
    |u_+(z, \bx)| = |u_-(z, \bx)| = 1,\qquad \forall\,(z, \bx)\in (0, \infty)\times \Omega,
\end{align*}
yet their initial conditions 
\begin{equation*}
    f_+(\bx) := u_+(0, \bx) = e^{i\psi(\bx)},\qquad f_-(\bx) := u_-(0, \bx) = e^{-i\psi(\bx)}
\end{equation*}
belong to different equivalence classes.
\end{proof}

A similar construction as in~\Cref{PROP:General Nonunique} also works here.
\begin{proposition}
For any real-valued $\kappa(z,\bx)=\kappa(\bx)$ and $\beta(z, \bx)=\beta(\bx)$ such that $\dfrac{\beta(\bx)}{\kappa(\bx)}\neq 0$, let $\varphi(\bx)\in \cC^2(\Omega; \bbR)$ be a real-valued function that solves
\[
    \Delta \varphi+\frac{2\beta}{\kappa} |\varphi|^2\varphi=0,\qquad \mbox{in}\ \ \Omega\,.
\]
We define
\[ v(\bx)=\varphi(\bx)+i\varphi(\bx)\,.
\]
Then 
\[
    u_+(z, \bx)=v(\bx) \quad \mbox{and}\quad u_-(z, \bx)=\overline{v(\bx)}
\]
are solutions to~\eqref{EQ:GP} with incident waves $v$ and $\overline{v}$ respectively, and $|u_+|=|u_-|$.
\end{proposition}

\section{Details on numerical implementations}
\label{SEC:Implementation}

We outline here the main computational procedures we adopted for the numerical experiments described in \Cref{SEC:Generalizations}. The main task is to solve the PDE-constrained optimization problem
\begin{equation}\label{EQ:Min}
    \min_{f} \Phi(f)
\end{equation}
subject to, $1\le m\le N_m$,
\begin{equation}\label{EQ:Min Constraints}
    \begin{array}{rcll}
    i\dfrac{\partial u_m}{\partial z} + \kappa_m \Delta u_m + \beta_m u_m^2 &=&0, & \mbox{in}\ \ (0, L_z)\times (0, L_x)\times(0, L_y)\\
    u_m(0, x,y) &=& f(x,y), & \text{in} \ \ (0, L_x)\times(0, L_y)\\
    u_m(z, 0,y) &=& u_m(z, L_x,y), &\text{on}\ \ (0, L_z)\times (0, L_y)\\
    u_m(z, x,0) &=& u_m(z, x, L_y),& \text{on}\ \ (0, L_z)\times (0, L_x).
    \end{array}
\end{equation}
where $\Phi(f)$ is the functional defined in~\eqref{EQ:Min Func}. 

The computational procedure we outlined here can handle, in exactly the same manner, the case when the quadratic nonlinear equation in~\eqref{EQ:Min Constraints} is replaced with the linear Schr\"{o}dinger model~\eqref{EQ:LS} or the Gross-Pitaevskii model~\eqref{EQ:GP}.

\subsection{Adjoint state gradient calculation}
\label{SUBSEC:Adjoint}

To compute the gradient of the objective function $\Phi(f)$ with respect to the incident wave $f$, we use the standard adjoint state method. A rigorous proof of the differentiability of $\Phi$ with respect to $f$ can be done in appropriate function spaces (such as those set up in~\cite{Cazenave-Book03,LaLuZh-arXiv23,LaOkSaSaTe-arXiv24}), but it is out of the scope of this work. Instead, we outline here only the calculation process.

We first differentiate $\Phi$ with respect to $f$ to have that
\begin{equation}\label{EQ:Frechet}
    \Phi'(f)[\delta f] = 2 \sum_{m=1}^{N_m} \sum_{p=1}^{N_p} \int_{\Omega} (|u_{m,p}|^2-d_{m,p}^2) \, \Re\left( \overline{u_{m,p}} \, u_{m,p}'(f)[\delta f]\right) \, d\bx\,.
\end{equation}
where $u_m'(f)[\delta f]$ denote the derivative of $u_m$ at $f$ in the direction $\delta f$. Meanwhile, we differentiate the solution $u_m$ to~\eqref{EQ:Min Constraints} with respect to $f$ to have
\begin{equation}\label{EQ:Deri of u}
    \begin{array}{rcll}
    i\dfrac{\partial u_m'}{\partial z} + \kappa_m \Delta u_m' + 2\beta_m u_m u_m' &=&0, & \mbox{in}\ \ (0, L_z)\times (0, L_x)\times(0, L_y)\\
    u_m'(0, x,y) &=& \delta f, & \text{in} \ \ (0, L_x)\times(0, L_y)\\
    u_m'(z, 0,y) &=& u_m'(z, L_x,y), &\text{on}\ \ (0, L_z)\times (0, L_y)\\
    u_m'(z, x,0) &=& u_m'(z, x, L_y),& \text{on}\ \ (0, L_z)\times (0, L_x)\,.
    \end{array}
\end{equation}

For $1\le m\le N_m$, let $u_m$ be the solution to \eqref{EQ:Min Constraints}, and $w_m$ be the solution to the adjoint problem
\begin{equation}\label{EQ:QNLS Adjoint}
    \begin{array}{rcll}
    -i\dfrac{\partial w_m}{\partial z} + \kappa_m \Delta w_m + 2\beta_m u_mw_m &=&Q_m(z, x, y), & \mbox{in}\ \ (0, L_z)\times (0, L_x)\times(0, L_y)\\
    w_m(L_z, x,y) &=& 0, & \text{in} \ \ (0, L_x)\times(0, L_y)\\
    w_m(z, x,y) &=& 0, &\text{on}\ \ (0, L_z)\times \partial \Omega\\
    \bn\cdot \nabla w_m(z, x,y) &=& 0,& \text{on}\ \ (0, L_z)\times\partial\Omega
    \end{array}
\end{equation}
with
\[
Q_m(z, x, y)=2\dsum_{p=1}^{N_p} (|u_{m,p}|^2-d_{m,p}^2)\overline{u_m}\delta(z-z_p)\,.
\]
Multiplying~\eqref{EQ:Deri of u} by $w_m$, multiplying~\eqref{EQ:QNLS Adjoint} by $u_m'$, integrating over $(0,L_z)\times(0, L_x)\times (0, L_y)$, and taking the difference of the results then gives us
\[
\sum_{p=1}^{N_p} \int_{\Omega} (|u_{m,p}|^2-d_{m,p}^2) \, \left( \overline{u_{m,p}} \, u_{m,p}'(f)[\delta f]\right) \, d\bx
=i\int_{\Omega} w_m(0,\bx) \delta f(\bx) \,d\bx.    
\]
We have therefore shown that the Fr\'echet derivative of $\Phi(f)$ at $f$ in the direction $\delta f$ is given by
\begin{equation}\label{EQ:Gradient QNLS}
    \Phi'(f)[\delta f] = -\sum_{m=1}^{N_m} \int_{\Omega} \Big[\Im( w_m(0, \bx)) \Re (\delta f)+ \Re( w_m(0, \bx)) \Im (\delta f)\Big]\,d\bx.
\end{equation}
\RED{Additionally, recall from~\eqref{EQ:Psi} that when the amplitude of $f$ is given (as say $a$), the problem reduces to minimizing $\Psi(\varphi) = \Phi(a\exp(i\varphi))$. In this case, the Fr\'echet derivative of $\Psi(\varphi)$ at $\varphi$ in the direction $\delta\varphi$ can be computed from the Fr\'echet derivative~\eqref{EQ:Gradient QNLS} of $\Phi$ and the chain rule:
\begin{align*}
    \Psi'(\varphi)[\delta\varphi] &= \Phi'(a\exp(i\varphi))[ia\exp(i\varphi)\delta\varphi]\\
    &= -\sum_{m=1}^{N_m} \int_{\Omega} \Big[\Im( w_m(0, \bx)) \Re (ia\exp(i\varphi)\delta\varphi)+ \Re( w_m(0, \bx)) \Im (ia\exp(i\varphi)\delta\varphi)\Big]\,d\bx\\
    &= -\sum_{m=1}^{N_m} \int_{\Omega} a\Big[-\Im( w_m(0, \bx))\sin\varphi + \Re( w_m(0, \bx))\cos\varphi\Big]\delta\varphi\,d\bx.
\end{align*}
}
The same calculations can be performed to solve the phase retrieval problem for the Gross-Pitaevskii equation. Indeed, we replace the quadratic nonlinear Schr\"{o}dinger model in the optimization problem~\eqref{EQ:Min} with the Gross-Pitaevski equation with periodic boundary conditions:
\begin{equation}\label{EQ:Min Constraints GP}
    \begin{array}{rcll}
    i\dfrac{\partial u_m}{\partial z} + \kappa_m \Delta u_m + \beta_m |u_m|^2u_m &=&0, & \mbox{in}\ \ (0, L_z)\times (0, L_x)\times(0, L_y)\\
    u_m(x,y,0) &=& f(x,y), & \text{in} \ \ (0, L_x)\times(0, L_y)\\
    u_m(z, 0,y) &=& u_m(z, L_x,y), &\text{on}\ \ (0, L_z)\times (0, L_y)\\
    u_m(z, x,0) &=& u_m(z, x, L_y),& \text{on}\ \ (0, L_z)\times (0, L_x).
    \end{array}
\end{equation}
for $1\le m\le N_m$, and introduce $w_m$ as the solution to the adjoint problem
\begin{equation}\label{EQ:GP Adjoint}
    \begin{array}{rcll}
    -i\dfrac{\partial w_m}{\partial z} + \kappa_m \Delta w_m + 2\beta_m |u_m|^2 w_m +\beta_m\overline{u_m}^2\overline w_m &=&Q_m(z, x, y), & \mbox{in}\ \ (0, L_z)\times (0, L_x)\times(0, L_y)\\
    w_m(L_z, x,y) &=& 0, & \text{on} \ \ (0, L_x)\times(0, L_y)\\
    w_m(z, x,y) &=& 0, &\text{on}\ \ (0, L_z)\times \partial \Omega\\
    \bn\cdot \nabla w_m(z, x,y) &=& 0,& \text{on}\ \ (0, L_z)\times\partial\Omega
    \end{array}
\end{equation}
with $Q_m$ defined in the same way as above. Then the Fr\'echet derivative of $\Phi(f)$ at $f$ in the direction $\delta f$ is again given as in~\eqref{EQ:Gradient QNLS}.

\Cref{ALGO:Grad} summarizes the procedure of gradient computation for the objective functions.
\begin{algorithm}[!htb]
\caption{Gradient Computation by Adjoint State Method}
\label{ALGO:Grad}
\begin{algorithmic}[1]
  \State Initialize $\Phi = 0$, $\bG = 0+0i$
  
  \For{$m=1,\cdots,N_m$}
      \State Solve ~\eqref{EQ:Min 
      Constraints} (resp. ~\eqref{EQ:Min Constraints GP} for GP model) for $u_m$ 
      \State Evaluate the residual and adjoint source
      \[
  r_{m,p}:= |u_m(z_p, \bx)|^2-d_{m,p}^2( \bx),\qquad 
  Q_{m,p}= 2r_{m,p} \, \overline{u_{s}(z_p, \bx)}
  \]
       \State Update
  $$\Phi = \Phi + \frac{1}{2} \sum_{m=1}^{N_p} \int_{\Omega} r_{m,p}^2(\bx)\, d\bx$$
    \State Solve the adjoint equation~\eqref{EQ:QNLS Adjoint} (resp.~\eqref{EQ:GP Adjoint} for GP model) for $w_m$
  \State Update
      $$\Re (\bG) = \Re (\bG)  -   \Im(w_m(0, \bx))$$
      $$\Im (\bG) = \Im (\bG)  -   \Re(w_m(0, \bx))$$
  \EndFor
  \State \Return $\Phi$, $\bG$
\end{algorithmic}
\end{algorithm}

\subsection{Domain discretization}

Our simulation covers the domain $(0,L_z)\times \Omega$ with $\Omega=(0, L_x)\times (0, L_y)$. We discretize $\Omega$ by $100\times 100$ cells of uniform size. The spatial grid sizes are $dx=dy=0.01$, and for numerical stability and resolution considerations, we choose $dz=1\times e^{-5}$. The total simulation time is fixed at $L_z=0.02$.  This time frame allows the waves to travel properly within the domain without using too many computational resources.

\subsection{Finite difference discretizations of PDEs}

The forward models~\eqref{EQ:Min Constraints} and~\eqref{EQ:Min Constraints GP}, as well as the corresponding adjoint equations~\eqref{EQ:QNLS Adjoint} and~\eqref{EQ:GP Adjoint} are all discretized using a standard time-domain finite difference method. In particular, we use the fourth-order Runge-Kutta method for the time variable. 

\subsection{Numerical optimization method}

To solve the minimization problem, we implemented a quasi-Newton method with the L-BFGS rule to update the approximated Hessian~\cite{NoWr-Book06}. Gradients of the objective function are computed with the adjoint statement method we outlined in~\Cref{SUBSEC:Adjoint}. For our Python-based implementation, we utilized the \texttt{scipy.optimize} toolkit for optimization routines. We enforce stopping criteria based on the magnitude of the objective function $\|\Phi\|$, the norm of the gradient $\|\nabla \Phi\|$, a pre-defined maximum iteration count, and the maximum number of line searches in each iteration.


\end{document}